\documentclass[twoside]{IEEEtran}

\usepackage{amsfonts,amssymb}
\usepackage[usenames,dvipsnames,svgnames,table]{xcolor}

\usepackage[font=small,labelfont=bf]{caption}

\usepackage{vkmacros}
\usepackage{fitin}
\usepackage{multibib}
\newcites{apx}{Appendix References}

\newif\ifmapx
\mapxfalse

\newif\ifshowtodo
\showtodofalse

\newif\ifshowdetail
\showdetailfalse

\newif\ifvlong
\vlongfalse

\long\def\vlong#1#2{\ifvlong{#1}\else{#2}\fi}

\usepackage{environ}

\NewEnviron{apxonly}
{
 \ifmapx\expandafter
 \addtocounter{equation}{-1}
\begin{subequations}\renewcommand\theequation{\theparentequation\alph{equation}}
  \begingroup\color{violet!80!black}
 \BODY
 \endgroup
\end{subequations}
 \fi
}

\renewcommand{\vect}[1]{{ #1}}
\newcommand{\lqr}[1]{ \mathrm{LQR} \left(#1 \right) } 
\newcommand{\wmse}[1]{ \mathrm{WMSE} \left(#1 \right) }

\title{Rate-cost tradeoffs in control}

\author{
Victoria Kostina, Babak Hassibi \vspace{-1.5em}
\thanks{
The authors are with California Institute of Technology (e-mail: \href{mailto:vkostina@caltech.edu}{vkostina@caltech.edu}, \href{mailto:hassibi@caltech.edu}{hassibi@caltech.edu}). 

A part of this work was presented at the 54th Annual Allerton Conference on Communication, Control and Computing \cite{kostina2016control}. 

The work of Victoria Kostina was supported in part by the National Science Foundation (NSF)
under grants CCF-1566567 and CCF-1751356. The work of Babak Hassibi was supported in part by the NSF under grants CNS-0932428, CCF-1018927, CCF-1423663 and CCF-1409204, by a grant from Qualcomm Inc., by NASA's Jet Propulsion Laboratory through the President and Director's Fund, and by King Abdullah University of Science and Technology.

}

}

\IEEEoverridecommandlockouts
\begin{document}
\maketitle
\begin{abstract}
Consider a control problem with a communication
channel connecting the observer of a linear stochastic system to the controller. The goal of the controller is to minimize a quadratic cost function in the state variables and control signal, known as the linear quadratic regulator (LQR). We study the fundamental tradeoff between the communication rate $r$ bits/sec and the expected cost $b$. We obtain a lower bound on a certain rate-cost function, which quantifies the minimum directed mutual information between the channel input and output that is compatible with a target LQR cost. The rate-cost function has operational significance in multiple scenarios of interest: among others, it allows us to lower-bound the minimum communication rate for fixed and variable length quantization, and for control over noisy channels. We derive an explicit lower bound to the rate-cost function, which applies to the vector, non-Gaussian, and partially observed systems, thereby extending and generalizing an earlier explicit expression for the scalar Gaussian system, due to Tatikonda el al. \cite{tatikonda2004stochastic}.   
The bound applies as long as the differential entropy of the system noise is not $-\infty$. 
It can be closely approached by a simple lattice quantization scheme that only quantizes the innovation, that is, the difference between the controller's belief about the current state and the true state. Via a separation principle between control and communication, similar results hold for causal lossy compression of additive noise Markov sources. Apart from standard dynamic programming arguments, our technical approach leverages the Shannon lower bound, develops new estimates for data compression with coding memory, and uses some recent results on high resolution variable-length vector quantization to prove that the new converse bounds are tight.

\end{abstract}

\begin{IEEEkeywords}
Linear stochastic control, LQR control, remote control, rate-distortion tradeoff, high resolution, causal rate-distortion theory, Gauss-Markov source.
\end{IEEEkeywords}

\section{Introduction}

\subsection{System model}
Consider a discrete-time stochastic linear system: 
 \begin{align}
\vect{X}_{i+1} &=  \mathsf A \vect{X}_i + \mathsf B \vect{U}_i + \vect{V}_i, \label{eq:system1}
\end{align}
where $X_i \in \mathbb R^n$ is the state, $\vect{V}_i \in \mathbb R^n$ is the process noise, $\vect{U}_i \in \mathbb R^m$ is the control action, and $\mathsf A$  and $\mathsf B$ are fixed matrices of dimensions $n \times n$ and $n \times m$, respectively.
 At time $i$, the controller observes output $G_i$ of the channel, and chooses a control action $U_i$ based on the data it has observed up to time $i$.  %
At time $i$, the encoder observes the output of the sensor $Y_i \in \mathbb R^k$:
 \begin{align}
\vect{Y}_i &= \mathsf C X_i + W_i,  \label{eq:system2}
\end{align}
 where $\mathsf C$ is a $k \times n$ deterministic matrix, and $W_i \in \mathbb R^k$ is the observation noise. The encoder forms a codeword $F_i$, which is then passed through the channel. Like the controller, the encoder has access to the entire history of the data it has observed. %
 See Fig. \ref{fig:system}. 
 
We assume that system noises $\vect{V}_1, V_2 \ldots$ are i.i.d. zero-mean, that observation noises $W_1, W_2, \ldots$ are  zero-mean, i.i.d. independent of $\{W_i\}_{i=1}^\infty$, and that $X_1$ is zero-mean and independent of $\{V_i, W_i\}_{i=1}^\infty$. We make the usual assumption that the pair $(\mathsf A, \mathsf B)$ is controllable and that the pair $(\mathsf A, \mathsf C)$ is observable.
 If the encoder observes the full system state, i.e. $Y_i = X_i$ (rather than its noise-corrupted version as in \eqref{eq:system2}), then we say that the system is \emph{fully observed} (rather than \emph{partially observed}).

\vspace{5pt}
\begin{figure}[htp]
\begin{center}
    \epsfig{file=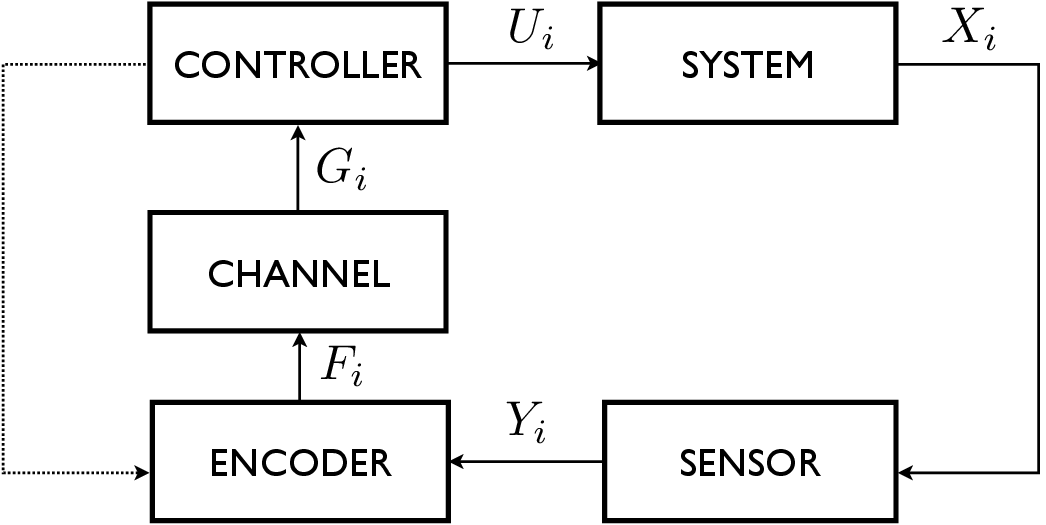,width=.7\linewidth}
\end{center}
 \caption[]{The distributed control system.} \label{fig:system}
\end{figure}

\emph{Notation:} Capital letters $X$, $Y$ denote (spatial) random vectors; $X^t \triangleq (X_1, \ldots, X_{t})$ denotes temporal random vectors, or the history of vector samples up to time $t$; $X_i^t \triangleq (X_i, \ldots, X_t)$ (empty if $t < i$); $X^\infty \triangleq ( X_1, X_2, \ldots )$; for i.i.d. random vectors, $X$ denotes a random vector distributed the same as each of $X_1, X_2, \ldots$; $\mathcal D$ represents a delay by one, i.e. $\mathcal D X^{t} \triangleq (0, X^{t-1})$; $\mathsf \Sigma_X \triangleq \E{(X - \E{X}) (X- \E{X})^T}$ denotes the covariance matrix of random vector $X$; $X \pperp Y$ reads ``$X$ is independent of $Y$''.  Sans-serif capitals $\mat A$, $\mat B$ denote  constant matrices; $\mat A \succeq \mat B$ ($\succ \mat B$) signifies that $\mat A - \mat B$ is positive semidefinite (definite); $\mat I_n$ is the $n \times n$ identity matrix; $\mat 0_{k \times n}$ is the $k \times n$ all-zero matrix. Lowercase letters $a, b, \ldots$ denote known scalars. $\|\cdot \|$ denotes the Euclidean norm, $\triangleq$ reads ``by definition''; $|\cdot|_+ \triangleq \max\{0, \cdot\}$.

\subsection{The rate-cost tradeoff}

The efficiency of a given control law at time $t$ is measured by the linear quadratic regulator (LQR) cost function:\footnote{As common in information theory, here we abuse the notation slightly and  write $\lqr{X^{t}, U^{t-1}}$ to mean that $\lqr{}$ is a function of the joint distribution of $X^{t}, U^{t-1}$.}
\begin{flalign}
& \mathrm{L\hspace{-.5pt}Q\hspace{-.5pt}R}({X^{t} \!, U^{t-1} })  \label{eq:lqr}%
 \triangleq  \E{ \! \sum_{i = 1}^{t-1}\!\! \left( \vect X_{i}^T \mathsf Q X_{i} \!\! +\! \vect U_i^T \mathsf R U_i \right) \!\!+ \!\! \vect X_{t}^T \mathsf S_{t} X_{t}\! } \!\!  ,\hspace{-2em}& 
\end{flalign}
where $\mathsf Q \succeq \mat 0$, $\mathsf R \succeq \mat 0$ and $\mathsf S_{t} \succeq \mat 0$. The LQR cost balances between the deviation of the system from the desired state $\mathbf 0$ and the  control power, which are defined with respect to the norms induced by the matrices $\mathsf Q$ (and $\mathsf S_{t}$) and $\mathsf R$.
In the special case $\mathsf Q = \mathsf I_n$, $\mathsf R = \mathsf 0$ and $\mathsf S_{t + 1} = \mathsf I_n$, the cost function in \eqref{eq:lqr} is the average mean-square deviation of the system from $\mathbf 0$, $\E{\sum_{i = 1}^{t} \|\vect X_{i}\|^2}$. 

Given a joint distribution of random vectors $U^t$ and $Y^t$, the directed mutual information is defined as \cite{massey1990causality}
\begin{equation}
I(Y^t \to U^t) \triangleq \sum_{i = 1}^{t}  I(Y^i; U_i | U^{i-1}). \label{eq:Idir}
\end{equation}
Directed mutual information, which captures the information due to causal dependence of $U^t$ on $Y^t$, and which is less than or equal to the full mutual information $I(Y^t; U^t)$, has proven useful in communication problems where causality and feedback play a role. Given a joint distribution $P_{Y^t U^t}$, it is enlightening to consider \emph{causally conditional} probability kernel \cite{kramer1998PhD}
\vspace{-1em}
\begin{equation}
 P_{U^t || Y^t} \triangleq \prod_{i=1}^{t} P_{U_i | U^{i-1}, Y^i} \label{eq:causalcond}
\end{equation}
Note that $P_{Y^t U^t} = P_{Y^t \| \mathcal D U^{t}} P_{U^t \| Y^t}$.
In \figref{fig:system}, the system dynamics \eqref{eq:system1}, \eqref{eq:system2} fixes the kernels $P_{Y^t \| \mathcal D U^{t}}$, $t = 1, 2, \ldots$, while the causal channels $P_{U^t \| Y^t}$ comprise the encoder, the channel and the controller. 

The following information-theoretic quantity will play a central role in determining the operational fundamental limits of control under communication constraints.  
\begin{defn}[rate-cost function]
The rate-cost function of the dynamical system $\{P_{Y^t \| \mathcal D U^{t}}\}_{t=0}^\infty$ is defined as
\begin{equation}
\mathbb R(b) \triangleq \limsup_{t \to \infty}  \\
 \inf_{\substack{ P_{U^t \| Y^t}  \colon \\
\frac 1 t \lqr{X^{t},\, U^{t-1}} \leq b}} 
 \frac 1 t I(Y^t \to U^t) 
 \label{eq:Rbo}
\end{equation}
\label{defn:rb}
\end{defn}
\vspace{-1em}
 In this paper, we will show a simple lower bound to the rate-cost function \eqref{eq:Rbo} of the stochastic linear system \eqref{eq:system1}, \eqref{eq:system2}. Although $\mathbb R(b)$ does not have a direct operational interpretation unless the channel is probabilistically matched~\cite{gastpar2003tocodeornot} to the system, it is 
linked to the minimum data rate required to keep the system at LQR cost $b$, over both noiseless and noisy channels. Namely, we will show that $\mathbb R(b)$ provides a lower bound on the minimum capacity of the channel necessary to sustain LQR cost $b$, valid for any encoder/controller pair. We will also show that over noiseless channels, $\mathbb R(b)$ can be closely approached by a simple variable-length lattice-based quantization scheme that transmits only the innovation. 

\subsection{Prior art}

The analysis of control under communication constraints has a rich history.
 The first results on the minimum data rate required for stabilizability appeared in \cite{baillieul1999feedback,wong1999systems}.  These works analyze the evolution of a scalar system from a worst-case perspective. In that setting, the initial state $X_1$ is assumed to belong to a bounded set, the process noise $V_1, V_2, \ldots$ is assumed to be bounded, and the system is said to be {\it stabilizable} if there exists a (rate-constrained) control sequence such that the worst-case deviation of the system state from the target state $\mathbf 0$ is bounded: $\limsup_{t \to \infty} \| X_t\| < \infty$. In \cite{baillieul1999feedback,wong1999systems}, it was shown that a fully observed unstable scalar system can be kept bounded by quantized control if and only
if the available data rate exceeds $\log \mat A$ bits per sample. Tatikonda and Mitter \cite{tatikonda2004control} generalized this result to vector systems; namely, they showed that the necessary data rate to stabilize a vector system with bounded noise is at least
\begin{equation}
r > \sum_{i \colon |\lambda_i(\mathsf A)| \geq 1} \log |\lambda_i(\mathsf A)|, \label{eq:tatimi}
\end{equation}
where the sum is over the unstable eigenvalues of $\mathsf A$, i.e. those eigenvalues whose magnitude exceeds $1$. Compellingly, \eqref{eq:tatimi} shows that only the nonstable modes of $\mat A$ matter; the stable modes can be kept bounded at any arbitrarily small quantization rate (and even at zero rate if $V_t \equiv 0$).  
Using a volume-based argument, Nair et al. \cite{nair2007feedback} showed a lower bound to quantization rate in order to attain $\limsup_{t \to \infty} \| X_t\| \leq d$, thereby refining \eqref{eq:tatimi}. Nair et al. \cite{nair2007feedback} also presented an achievability scheme confirming that for scalar systems, that bound is tight.

Nair and Evans \cite{nair2004stabilizability} showed for systems with unbounded process and observation disturbances, Tatikonda and Mitter's condition on the rate \eqref{eq:tatimi} continues to be necessary and sufficient in order to keep the mean-square deviation of the plant state from $\mathbf 0$ bounded, that is, in order to satisfy $\limsup_{t \to \infty} \E{\| X_t\|^2} < \infty$.

Nair and Evans' converse bound \cite{nair2004stabilizability} applies to fixed-rate quantizers, that is, to compressors whose outputs can take one of $2^r$ values. Time-invariant fixed-rate quantizers are unable to attain bounded cost if the noise is unbounded \cite{nair2004stabilizability}, regardless of their rate. The reason is that since the noise is unbounded, over time, a large magnitude noise realization will inevitably be encountered, and the dynamic range of the quantizer will be exceeded by a large margin, not permitting recovery. Adaptive quantization schemes, which ``zoom out'' (i.e. stretch the quantization intervals) when the system is far from the target and  ``zoom in'' when the system is close to the target, are studied in \cite{brockett2000quantized,nair2004stabilizability,yuksel2010fixedrate,yuksel2014lqg}. Structural properties of optimal zero-delay quantizers for the compression of Markov sources were investigated in \cite{witsenhausen1979structure,gaarder1982optimal,walrand1983optimal,borkar2001optimal,teneketzis2006structure,linder2014zerodelay,wood2015walrand,yuksel2008ita}.

In variable-rate (or \emph{length}) quantization, the quantizer can have a countably infinite number of quantization cells.  \emph{Entropy coding} is applied to encode the indices of quantization cells, so that the more likely quantization cells have a shorter description and the less likely ones a longer one. %
Elia and Mitter \cite{elia2001stabilization} considered stabilization of a noiseless linear system controlled with a variable-length scalar quantizer, and showed that for a certain notion of {\it coarseness}, the coarsest quantizer has levels that follow a logarithmic law.   

Beyond worst-case and mean-square stabilizability, Tatikonda et al. \cite{tatikonda2004stochastic} considered a setting known as linear quadratic Gaussian (LQG) control (\eqref{eq:system1}, \eqref{eq:system2} with Gaussian disturbances and LQR cost function in \eqref{eq:lqr}) with communication constraints and tied the minimum attainable LQG cost to the Gaussian {\it causal rate-distortion function}, introduced decades earlier by Gorbunov and Pinsker \cite{gorbunov1974prognostic}, which is equal to the minimal (subject to a distortion constraint) directed mutual information between the stochastic process and its quantized representation \cite{charalambous2014nonanticipative}. Stabilizability of LQG systems under a directed mutual information constraint was studied in \cite{charalambous2008LQGoptimality}. The problem of minimizing an arbitrary cost function in control of a general process under a directed mutual information constraint was formulated in \cite{charalambous2011memoryfeedback}. Control of general Markov processes under a mutual information constraint was studied in \cite{shafieepoorfard2016rationally}. Silva et al. \cite{silva2011framework} elucidated the operational meaning of directed mutual information, by pointing out that it lower-bounds the rate of a quantizer embedded into a feedback loop of a control system, and by showing that the bound is approached to within 1 bit by a {\it dithered} prefix-free quantizer, a compression setting in which both the compressor and the decompressor have access to a common dither - a random signal with special statistical properties. More recently, Silva et al.~\cite{silva2016characterization} computed a lower bound to the minimum quantization rate in scalar Gaussian systems with stationary disturbances and proposed a dithered quantization scheme that performs within 1.254 bits from it. Tanaka et al. \cite{tanaka2016isit} generalized the results of \cite{silva2016characterization} to vector systems. A connection between causal rate-distortion function and Kalman filtering using dithered variable-length quantizers was explored in \cite{stavrou2017ITWKalman}.

Causal rate-distortion function is challenging to evaluate, and beyond the scalar Gauss-Markov source \cite{gorbunov1974prognostic,tatikonda2004stochastic}, no closed-form expression is known for it. For stationary scalar Gaussian processes, Derpich and Ostergaard \cite{derpich2012uppercausal} showed an upper bound and Silva et al.~\cite{silva2016characterization} a lower bound. For vector Gauss-Markov sources, Tanaka et al. developed a semidefinite program to compute exactly the minimum directed mutual information in quantization~\cite{tanaka2017semidefinite} and control~\cite{tanaka2017lqg}.

\subsection{Our contribution}

In this paper, we show a lower bound to $\mathbb R(b)$ of a \emph{fully observed} system. We do not require the noise $V_i$ to be bounded or Gaussian.  We also show that \eqref{eq:tatimi} remains necessary to keep the LQR cost bounded, even if the system noise is non-Gaussian, generalizing previously known results. Although our converse lower bound holds for a general class of codes that can take full advantage of the memory of the data observed so far and that are not constrained to be linear or have any other particular structure, we show that the new bound can be closely approached within a much more narrow class of codes. Namely, a simple variable-length quantization scheme, which uses a lattice covering and which only transmits the difference between the controller's estimate about the current system state and the true state, performs within a fraction of a bit from the lower bound, with a vanishing gap as $b$ approaches its minimum attainable value, $b_{\min}$. The scheme is a variant of a classical differential pulse-code modulation (DPCM) scheme, in which a variable-length lattice code is used to encode the innovation process. Unlike previously proposed quantization schemes with provable performance guarantees, our scheme does not use dither. 

 Our results generalize to \emph{partially observed} systems, where the encoder does not have access to $X_i$ but only to its noise-corrupted version, $Y_i$. For those results to hold, we require the system and observation noises to be Gaussian. 
 
Our approach is based on a new lower bound to causal rate-distortion function, termed the \emph{causal Shannon lower bound} (\thmref{thm:seqslb} in \secref{sec:causalRd} below), which holds for  vector Markov sources with continuous additive disturbances, as in \eqref{eq:system1}.  For the scalar Gauss-Markov source, the bound coincides with a previously known expression \cite{gorbunov1974prognostic}. %

\subsection{Technical approach}

The main idea behind our approach to show a converse (impossibility) result is to recursively lower-bound distortion-rate functions arising at each step. We apply the classical Shannon lower bound  \cite{shannon1959coding}, which bounds the distortion-rate function $X$ in terms of the entropy power of $X$, and we use the entropy power inequality \cite{shannon1948mathematical,stam1959some} to split up the distortion-rate functions of sums of independent random variables. Since Shannon's lower bound applies regardless of the distribution of the source random variable, our technique circumvents a precise characterization of the distribution of the state at each time instant. The technique also does not restrict the system noises to be Gaussian.  

To show that our bound can be approached at high rates, we build on the ideas from high resolution quantization theory. A pioneering result of Gish and Piece \cite{gish1968asymptotically} states that in the limit of high resolution, a uniform scalar quantizer incurs a loss of only about $\frac 1 2 \log_2 \frac{2 \pi e}{12} \approx 0.254$ bits per sample.  Ziv \cite{ziv1985universal} showed that regardless of target distortion, the normalized output entropy of a dithered scalar quantizer exceeds that of the optimal vector quantizer by at most $\frac 1 2 \log \frac {4 \pi e} {12} \approx 0.754$ bits per sample. A lattice quantizer presents a natural extension of a scalar uniform quantizer to multiple dimensions. The advantage of lattice quantizers over uniform scalar quantizers is that the shape of their quantization cells can be made to approach a Euclidean ball in high dimensions \cite{zamir1996onlatticenoise}. %
Furthermore, the entropy rate of dithered lattice quantizers  converges to Shannon's lower bound in the limit of vanishing distortion \cite{gersho1979asymptotically,zamir1992universal,linder1994tessellating}. 

While the presence of a dither signal both at the encoder and the decoder greatly simplifies the analysis and can improve the quantization performance, it also complicates the engineering implementation. In this paper, we do not consider dithered quantization. Neither do we rely directly on the classical heuristic reasoning by Gish and Piece \cite{gish1968asymptotically}. Instead, we use a non-dithered lattice quantizer followed by an entropy coder. To rigorously prove  that its performance approaches our converse bound, we employ a recent upper bound \cite{kostina2016lowd} on the output entropy of lattice quantizers in terms of the differential entropy of the source, the target distortion and a smoothness parameter of the source density.

\subsection{Paper organization}
In \secref{sec:main}, we state and discuss our main results: \secref{sec:fo} focuses on the scenario where the observer sees the system state (fully observed system), \secref{sec:po} discusses a generalization to the scenario where the observer sees a noisy measurement of the system state (partially observed system), and \secref{sec:oper} discusses the operational implications of our bounds in 
the settings of fixed-rate quantization, variable-rate quantization and joint source-channel coding. 
In \secref{sec:causalRd}, we introduce the causal lossy compression problem, and we state the causal Shannon lower bound, together with a matching achievability result. In \secref{sec:sep}, we discuss separation between control, estimation and communication, a structural result that allows us to disentangle the three tasks. 
The proofs of the converse results are given in \secref{sec:c}, and the achievability schemes are presented in \secref{sec:a}.

\section{Main results}
\label{sec:main}

\subsection{Fully observed system}
\label{sec:fo}

In the absence of communication constraints, the minimum LQR cost attainable in the limit of infinite time is:
\begin{equation}
 b_{\min} = %
  \tr(\mathsf \Sigma_V \mathsf S), \label{eq:cmin}
\end{equation}
where $\mathsf \Sigma_V$ is the covariance matrix of each of the $V_1, V_2, \ldots$, and $\mathsf S$ is the solution to the algebraic Riccati equation
\begin{align}
\mathsf S &= \mathsf Q + \mathsf A^T \left( \mathsf S -  \mathsf M  \right) \mathsf A,\label{eq:S}\\
\mathsf M &\triangleq \mat L^T ( \mathsf R + \mathsf B^T \mathsf S \mathsf B) \mat L%
= \mathsf S \mathsf B  ( \mathsf R + \mathsf B^T \mathsf S \mathsf B)^{-1}  \mathsf B^T \mathsf S \label{eq:M}, \\
\mathsf L &\triangleq 
( \mathsf R + \mathsf B^T \mathsf S \mathsf B)^{-1} \mathsf B^T \mathsf S. \label{eq:L}
\end{align}

Our results quantifying the overhead over \eqref{eq:cmin} due to communication constraints are expressed in terms of the entropy power of the system and observation noises.  
The entropy power of an $n$-dimensional random vector $X$ is \footnote{All $\log$'s and $\exp$'s are common arbitrary base specifying the information units.}
\begin{equation}
 N(X) \triangleq \frac 1 {2 \pi e} \exp\left( \frac 2 n h(X) \right), \label{eq:ep}
\end{equation}
where $h(X) = -\int_{\mathbb R^n} f_X(x) \log f_X(x) dx$ is the differential entropy of $X$, and $f_X(\cdot)$ is the density of $X$ with respect to the Lebesgue measure on $\mathbb R^n$. 
The entropy power satisfies the following classical inequalities:  
\begin{align}
 N(X)  &\leq \left( \det \mathsf \Sigma_X \right)^{\frac 1 {n}} \leq \frac 1 {n} \Var{X}. \label{eq:epgauss} 
\end{align}
The first equality in \eqref{eq:epgauss} is attained if and only if  
 $X$ is Gaussian and the second if and only if $X$ is white.

Our first result is a lower bound on the rate-cost function. 
\begin{thm}
\label{thm:main} 
Consider the fully observed linear stochastic system \eqref{eq:system1}. Suppose that $h(V) > -\infty$.  At any LQR cost $b > \tr(\mathsf \Sigma_V \mathsf S)$, the rate-cost function is bounded below as 
\begin{align}
\!\!\!\! \mathbb R(b)  \geq \label{eq:main} %
&~  \log |\det \mathsf A |+ \frac {n} 2 \log \left( 1 + \frac{ N(V) |\det \mathsf M |^{\frac 1 n}}{(b - \tr(\mathsf \Sigma_V \mathsf S))/n }\right)\!\!, \!
\end{align}
where $\mathsf M$ is defined in \eqref{eq:M}. 
\end{thm}
The bound in \thmref{thm:main} is nontrivial if $\mathsf M \succ \mathsf 0$, which happens if $\rank \mathsf B = n$ and either $\mathsf Q \succ \mathsf 0$ or $\mathsf R \succ \mathsf 0$.
The bound in Theorem \ref{thm:main} continues to hold whether or not at time $i$ the encoder observes the previous control inputs $U_1, U_2, \ldots, U_{i-1}$. 

The right-hand side of \eqref{eq:main} is a decreasing function of $b$, which means that the controller needs to know more information about the state of the system to attain a smaller target cost.  
As an important special case, consider the rate-cost tradeoff where the goal is to minimize the mean-square deviation from the desired state $\mathbf 0$. Then, $\mathsf Q = \mathsf I_{n}$, $\mathsf R = \mathsf 0$, $\mathsf S = \mathsf M = \mathsf I_n$, 
$ b_{\min} = \Var{V}$, 
and \eqref{eq:main} particularizes as
\begin{align}
\!\!\! \mathbb R(b)  &\geq \log |\det \mathsf A | + \frac {n} 2 \log \left( 1 + \frac{N(V)}{(b - \Var{V})/n }\right) 
\label{eq:maincor}.
\end{align}
In another important special case, namely Gaussian $V$, \eqref{eq:main} particularizes as
\begin{align}
\!\!\! \mathbb R(b)  \geq  %
&~  \log |\det \mathsf A |+ \frac {n} 2 \log \left( 1 + \frac{  |\det \mathsf \Sigma_{V} \mathsf M |^{\frac 1 n}}{(b - \tr(\mathsf \Sigma_V \mathsf S))/n }\right)\!\!. \!\! \label{eq:maing}
\end{align}
For the scalar Gaussian system, \eqref{eq:maing} holds with equality. This is a consequence of known analyses \cite{gorbunov1974prognostic},  \cite{tatikonda2004stochastic}, \cite[Th.~3]{derpich2012uppercausal} (see also Remarks \ref{rem:gm} and \ref{rem:steady} in \secref{sec:causalRd} below). %

A typical behavior of \eqref{eq:maincor} is plotted in \vlong{\figref{fig:rdg} and }{}\figref{fig:rdl} as a function of target cost $b$. As $b \downarrow b_{\min}$, the required rate $\mathbb R(b) \uparrow \infty$. Conversely, as $b \uparrow \infty$, the rate  monotonically decreases and approaches $\log |\det \mathsf A|$. 
The rate-cost tradeoff provided by Theorem \ref{thm:main} can serve as a gauge for choosing an appropriate communication rate in order to meet the control objective. For example, in the setting of Fig. \ref{fig:rdl}, decreasing the data rate below 1 nat per sample incurs a massive penalty in cost, because the bound is almost flat in that regime. On the other hand, increasing the rate from 1 to 3 nats per sample brings a lot of improvement in cost, while further increasing it beyond $3$ nats results in virtually no improvement. 

Also plotted in \vlong{\figref{fig:rdg} and }{}\figref{fig:rdl} is the output entropy of a variable-rate uniform scalar quantizer that takes advantage of the memory of the past only through the innovation, i.e. the difference between the controller's prediction of the state at time $i$ given the information the controller had at time $i-1$ and the true state (see Section \ref{sec:a} for a precise description of the quantizer).    Its performance is strikingly close to the lower bound, being within $0.5$ nat even at large $b$, despite the fact that quantizers in this class cannot attain the optimal cost  exactly \cite{fu2012lack}. The gap further vanishes as $b$ decreases. The gap can be further decreased for multidimensional systems by taking advantage of lattice quantization. These effects are formally captured by the achievability result we are about to present, Theorem \ref{thm:maina}.

\vlong{
\begin{figure}[htp]
\begin{center}
    \epsfig{file=rdGaussian.eps,width=1\linewidth}
\end{center}
 \caption[]{The minimum quantizer entropy compatible with cost $b$ in a fully observed LQG system \eqref{eq:system1} with parameters $n = 1$, $\mathsf A = 2$, $\mathsf B = \mathsf Q = \mathsf R = 1$, $V \sim \mathcal N(0, 1)$.} \label{fig:rdg}
\end{figure}
}{}
\vspace{10pt}
\begin{figure}[htp]
\begin{center}
    \epsfig{file=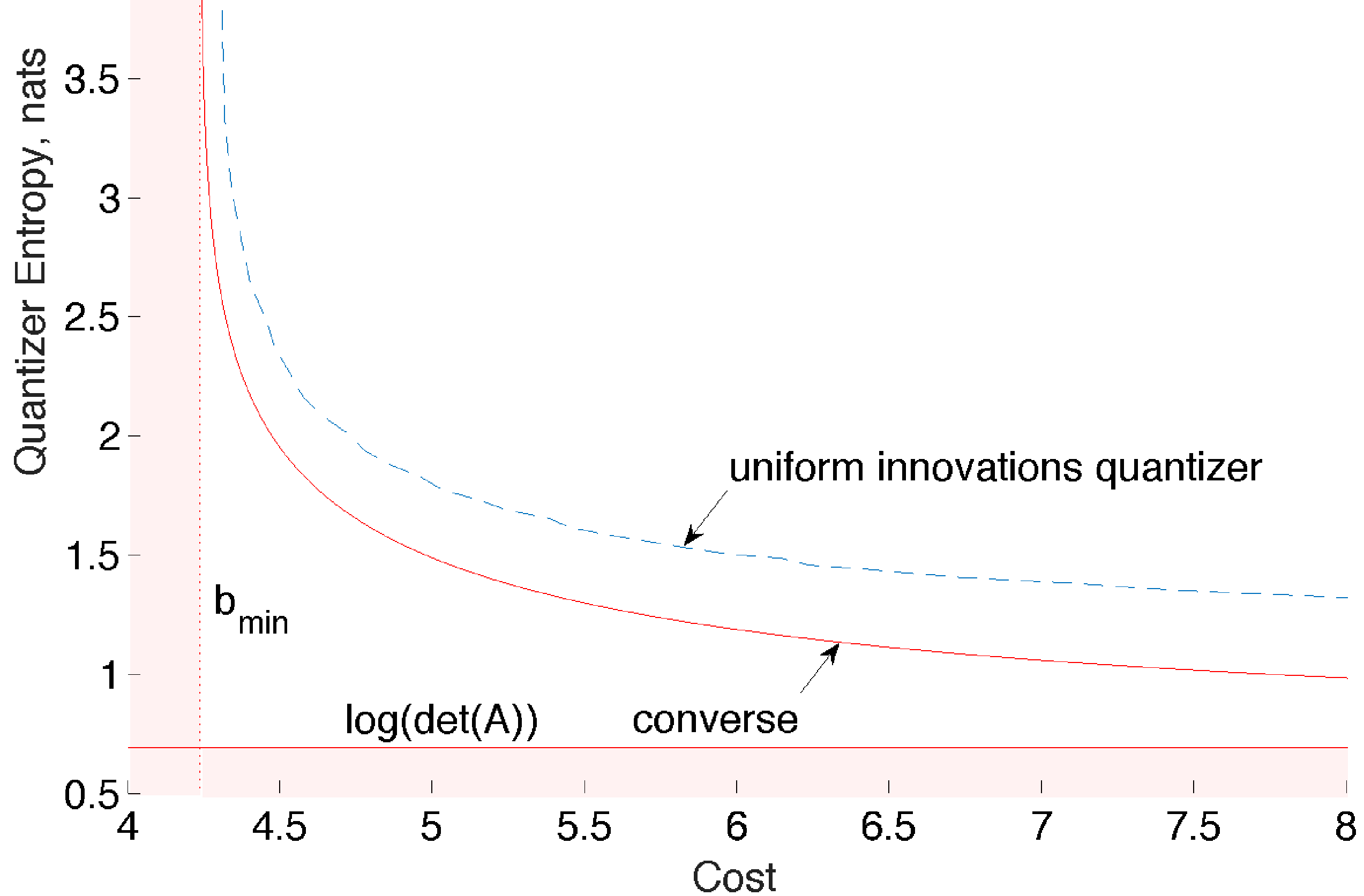,width=1\linewidth}
\end{center}
 \caption[]{The minimum quantizer entropy compatible with cost $b$ in a fully observed system \eqref{eq:system1} with parameters $n = 1$, $\mathsf A = 2$, $\mathsf B = \mathsf Q = \mathsf R = 1$, $V$ has Laplace distribution with variance 1.} \label{fig:rdl}
\end{figure}

Replacing in \defnref{defn:rb} the directed mutual information by the entropy of a causal quantizer, we introduce 
\begin{defn}[entropy-cost function]
\label{defn:hb}
The entropy-cost function of the dynamical system $\{P_{Y^t \| \mathcal D U^{t}}\}_{t=0}^\infty$ is defined as
\begin{equation}
\mathbb H(b) \triangleq \limsup_{t \to \infty}  \\
\inf_{\substack{ P_{U^t \| Y^t}  \colon \\
\frac 1 t \lqr{X^{t},\, U^{t-1}} \leq b}} 
H(U^t ). 
 \label{eq:hb}
\end{equation}
\end{defn}
Since $I(Y^i; U_i | U^{i-1})  \leq H(U_i | U^{i-1})$, we have
\begin{equation}
\mathbb H(b ) \geq \mathbb R(b). \label{eq:hlr}
\end{equation}
On the other hand, there exists a  variable-length quantizer that keeps the system at cost $b$  and whose average encoded length does not exceed $ \mathbb H(b)$ (see Section \ref{sec:oper} for details). Thus, unlike $\mathbb R(b)$, the function $\mathbb H(b)$ has a direct operational interpretation. %

Theorem \ref{thm:maina}, presented next, holds under the assumption that the density of the noise is sufficiently smooth. Specifically, we adopt the following notion of a regular density.
\begin{defn}[Regular density, \cite{polyanskiy2015wasserstein}]
\label{defn:pwreg}
Let $c_0 \geq 0$, $c_1 \geq 0$.
Differentiable probability density function $f_X$ of a random vector $X \in \mathbb R^n$ is called $(c_0, c_1)$-regular if \footnote{As usual, $\nabla$ denotes the gradient.}
\begin{equation}
\| \nabla f_X(x)\| \leq (c_1 \|x\| + c_0) f_X(x),  \qquad \forall x \in \mathbb R^n. %
\label{eq:pwreg}
\end{equation}
\end{defn}
A wide class of densities satisfying smoothness condition \eqref{eq:pwreg} is identified in \cite{polyanskiy2015wasserstein}. Gaussian, exponential, uniform, Gamma, Cauchy, Pareto distributions are all regular. Convolution with Gaussians produces a regular density: more precisely, the density of $B + Z$, with $B \pperp Z$ and $Z \sim \mathcal N(0, \sigma^2\, \mathsf I)$, is $(\frac {4}{\sigma^2} \E{\|B\|}, \frac {3}{\sigma^2})$-regular. Likewise, if the density of $Z$ is $(c_0, ~c_1)$-regular, then that of $B + Z$, where $\|B\| \leq b$ a.s., $B \pperp Z$ is $(c_0 + c_1 b,~ c_1)$-regular.

\begin{thm}
\label{thm:maina} 
Consider the fully observed linear stochastic system \eqref{eq:system1}, $Y_i = X_i$. Suppose that $\mathsf M \succ 0$ and that $V$ has a regular density.  
Then, at any LQR cost $b > \tr(\mathsf \Sigma_V \mathsf S))$, the entropy-cost function is bounded by
\begin{align}
\mathbb H(b)  \leq 
&~  \log |\det \mathsf A |+ \frac {n} 2 \log \left( 1 + \frac{ N(V) |\det \mathsf M|^{\frac 1 n}}{(b - \tr(\mathsf \Sigma_V \mathsf S) )/n }\right) \notag\\
+&~ O_1(\log n) + O_2\left( \left( b - \tr(\mathsf \Sigma_V \mathsf S) \right)^{\frac 1 2} \right),  \label{eq:maina} 
\end{align}
where $O_1(\log n) \leq C_1 \log n$ and $O_2(\xi) \leq C_2 \min \left\{ \xi, c_2 \right\}$ for some nonnegative constants $C_1,\ C_2$ and $c_2$. %
\end{thm}

The first two terms in \eqref{eq:maina} match the first two terms in \eqref{eq:main}. The $O_1(\log n)$ term is the penalty due to the shape of lattice quantizer cells not being exactly spherical, i.e. it is the penalty due to the space-filling loss of the quantizer at finite $n$. In \secref{sec:a}, we provide a precise expression for that term for $n = 1, 2, \ldots$.  The $O_2\left( \left( b - \tr(\mathsf \Sigma_V \mathsf S) \right)^{\frac 1 2} \right)$ is the penalty due to the distribution of the innovation not being uniform. It becomes negligible for small $b - \tr(\mathsf \Sigma_V \mathsf S) $, and the speed of that convergence depends on the smoothness parameters of the noise density.

\thmref{thm:maina} implies that if the channel $F_i \to G_i$ is noiseless, then there exists a quantizer with output entropy given by the right side of \eqref{eq:maina} that attains LQR cost $b> \tr(\mathsf \Sigma_V \mathsf S))$, when coupled with an appropriate controller. In fact, the bound in \eqref{eq:maina} is attainable by a simple lattice quantization scheme that only transmits the innovation of the state (a DPCM scheme). The controller computes the control action based on the quantized data as if it was the true state (the so-called \emph{certainty equivalence} control). 

 Theorem \ref{thm:main} gives a lower (converse) bound on the output entropy of quantizers that achieve the target cost $b$, without making any assumptions on the quantizer structure and permitting the use of the entire history of observation data. Theorem \ref{thm:maina} proves that the converse can be approached by a strikingly simple quantizer coupled with a standard controller, without common randomness (dither) at the encoder and the decoder.  Furthermore, although nonuniform rate allocation across time is allowed by \defnref{defn:hb}, such freedom is not needed to achieve \eqref{eq:maina}; the scheme that achieves \eqref{eq:maina} satisfies $H(U_i |U^{i-1} ) \to r$ in the limit of large~$i$.

Although the bound in \thmref{thm:main} is tight at low $b$ (as demonstrated by \thmref{thm:maina}), if $\mathsf A$ has both stable and unstable eigenvalues and $b$ is large, it is possible to improve the bound in \thmref{thm:main} by projecting out the stable modes of the system. Towards this end, consider a Jordan decomposition of $\mathsf A$:\footnote{A Jordan decomposition of $\mat A$ has been previously applied in the context of control under communication constraints in e.g. \cite{tatikonda2004control,nair2004stabilizability}.} 
\begin{equation}
 \mathsf A = \mathsf J {\mathsf A^\prime}  \mathsf J^{-1}, \label{eq:Aprime}
\end{equation}
where ${\mathsf A}^\prime$ is the Jordan form of $\mat A$, and $\mat J$ is invertible. 
Without loss of generality, assume the eigenvalues of $\mat A^\prime$ are ordered in decreasing magnitude order. Write $\mat A^\prime$ as a direct sum of its Jordan blocks, $\mat A^\prime = \mat A_{1} \oplus \ldots \oplus \mat A_{\bar s}$. For some $s$ such that $1 \leq s \leq \bar s \leq n$, let $\ell = \dim(\mat A_1) + \ldots + \dim(\mat A_s)$ be the dimension of the column space of the first $s$ Jordan blocks. Consider the  orthogonal projection matrix onto that space, $\mat \Pi_\ell \mat \Pi_\ell^T$, where $\mat \Pi_{\ell}$ is a 0-1-valued $n \times \ell$ matrix given by:
\begin{equation}
 \mat \Pi_{\ell} \triangleq 
\begin{bmatrix}
\mat I_{\ell} \\
\mat 0
 \end{bmatrix}
 \label{eq:pil}
\end{equation}

The improvement of \thmref{thm:main} can now be formulated. 
\begin{thm}
\label{thm:mainu}
Consider the fully observed linear stochastic system \eqref{eq:system1}, $Y_i = X_i$. Suppose that $h(V) > -\infty$.  Let $\mat \Lambda$ be a diagonal matrix such that $\mat M^\prime \triangleq \mathsf J^{T} \mathsf M \mathsf J \succeq \mat \Lambda$, where $\mat J$ is defined in \eqref{eq:Aprime}. At any LQR cost $b >  \tr(\mathsf \Sigma_V \mathsf S)$, the rate-cost function is bounded below as
\begin{align}
 \mathbb R(b) \geq \ell \log a^\prime + \frac \ell 2 \log \left( 1 + \frac{\mu^\prime N(\mat \Pi_{\ell}^T \mat J^{-1} V)}{(b - \tr(\mathsf \Sigma_V \mathsf S))/{\ell} }\right).\label{eq:mainpi} 
\end{align}
where $\mat \Pi_{\ell}$ is defined in \eqref{eq:pil}, and
\begin{align}
a^\prime &\triangleq \left|  \det \left( \mathsf \Pi_{\ell}^T \mathsf J^{-1} {\mathsf A}  \mathsf J \mat \Pi_{\ell} \right) \right|^\frac{1}{\ell}, \label{eq:aprime} \\
\mu^\prime &\triangleq \left( \det \mat \Lambda \right)^{\frac 1 {\ell}}. \label{eq:muprime}
\end{align}
\end{thm}

If $\ell = n$, then $a^\prime = \left|\det \mat A\right|$, and the bound in \eqref{eq:mainpi} reduces to \eqref{eq:main}. 
On the other hand, taking $\ell$ to be the number of unstable eigenvalues in \eqref{eq:mainpi}, one can conclude
\begin{equation}
\mathbb R(b) \geq \sum_{i \colon |\lambda_i(\mathsf A)| \geq 1} \log |\lambda_i(\mathsf A)|. \label{eq:mainu} 
\end{equation}

If  the dimensionality of control is less than that of the system, $m < n$, the bounds in \thmref{thm:main} and \thmref{thm:mainu} reduce to $\log |\det \mat A|$ and $\ell \log a^\prime$, respectively, losing the dependence on $b$. The bound in \thmref{thm:mainm} below is a decreasing function of $b$, even if $m < n$.

\begin{thm}
\label{thm:mainm}
Consider the fully observed linear stochastic system \eqref{eq:system1}, $Y_i = X_i$. Suppose that $h(V) > -\infty$.  At any LQR cost $b >  \tr(\mathsf \Sigma_V \mathsf S)$, the rate-cost function satisfies
\begin{flalign}
&\mathbb R(b) \!\geq\! \log |\det \mat A| \!+\! \frac m 2 \log \left( \frac{a^2}{|\det \mat A|^2} \!+\! \frac{\mu N(V)^{\frac n m} m}{b \!-\! \tr(\mathsf \Sigma_V \mathsf S) }\right)\!, \hspace{-5em} && \label{eq:maingen} 
\end{flalign}
\begin{align}
\hspace{-2em} \text{where} \hspace{3em} a &\triangleq \inf_{i \geq 1}  \left(\frac{ \det \left( \mathsf L \mathsf A^{i} \mathsf \Sigma_{V} \mathsf A^{i\, T} \mathsf L^{T} \right)}{\det \mathsf \Sigma_{V}  \det \left(\mat L \mat L^T\right)} \right)^{\frac 1 {2 i m}},\\
\mu &\triangleq \left( \det\left(\mat R + \mat B^T \mat S \mat B \right) \det \left(\mat L  \mat L^T\right)\right)^\frac{1}{m}, \label{eq:mugen}
\end{align}
and $\mat L$ is defined in \eqref{eq:L}.
\end{thm}
If $m = n$, \thmref{thm:mainm} reduces to \thmref{thm:main}.

We conclude \secref{sec:fo} with a few technical remarks.

\begin{remark}
Since the running mean-square cost is bounded above as
$
 \frac 1 {t}\E{\sum_{i = 1}^{t} \|\vect X_{i}\|^2} \leq \max_{1 \leq i \leq t} \E{\|X_i\|^2}
$,
our \thmref{thm:mainu} implies (via \eqref{eq:mainu}) the weaker result of Nair and Evans \cite{nair2004stabilizability}, who showed  the necessity of \eqref{eq:mainu} to keep $\sup_t \E{\|X_t\|^2}$ bounded. 
Note also that the approach of Nair and Evans \cite{nair2004stabilizability} applies only to fixed-rate quantization, while our approach encompasses both fixed- and variable-rate quantization, as well as control over noisy channels and non-Gaussian system disturbances.%
\end{remark}

\begin{remark}
Tatikonda et al. \cite[(17)]{tatikonda2004stochastic} proposed to apply the classical reverse waterfilling, known to achieve the noncausal Gaussian rate-distortion function, at each step $i$, to compute the causal Gaussian rate-distortion function for the sequence of vectors $\{X_i\}_{i=1}^{\infty}$.
Unfortunately,  reverse waterfilling is only suboptimal, as can be verified numerically by comparing \cite[(17)]{tatikonda2004stochastic} to the semidefinite program of Tanaka et al. \cite{tanaka2017semidefinite} (\figref{fig:rdl}). The reason reverse waterfilling does not apply is that $\lambda_i(t)$ in \cite[(15)]{tatikonda2004stochastic} depend on the distortion threshold, an effect not present in the computation of the classical non-causal rate-distortion function. %
\label{rem:tsm}
\end{remark}

\begin{remark}
The semidefinite program (SDP) of Tanaka et al. \cite{tanaka2017semidefinite} provides an exact numerical solution to the Gaussian causal rate-distortion-function, while our results in Theorems \ref{thm:main},~\ref{thm:mainu},~\ref{thm:mainm} provide analytical lower bounds, which hold beyond Gaussian noise. \figref{fig:rdl} presents a numerical comparison between our lower bound in \thmref{thm:mainu} and the exact calculation of $\mathbb R(d)$, for a randomly generated 3-dimensional system. For the example in \figref{fig:rdl}, our lower bound is within $0.14$ bits from the optimum. While always tight in low-cost regime, in medium-cost regime it will become looser if the spread of the eigenvalues of $\mat A$ is large. 
\end{remark}

\psfrag{R, bits}{\scriptsize{$\mathbb R(b)$, bits}}
\psfrag{SDP}{\scriptsize{SDP \cite{tanaka2017lqg}}}
\psfrag{Reverse waterfilling}{\scriptsize{Reverse waterfilling \cite{tatikonda2004stochastic}}}
\psfrag{lower bound}{\scriptsize{Lower bound (\thmref{thm:mainu})}}
\psfrag{d}{\scriptsize{$d$}}
\psfrag{Rlim}{\scriptsize{$\dfrac{C}{R(d)}$}}
\psfrag{Vjscc}{\scriptsize{$ \mathscr V(d, \alpha)$}}
\vspace{-.5em}
\begin{figure}[htp]
\begin{center}
    \epsfig{file=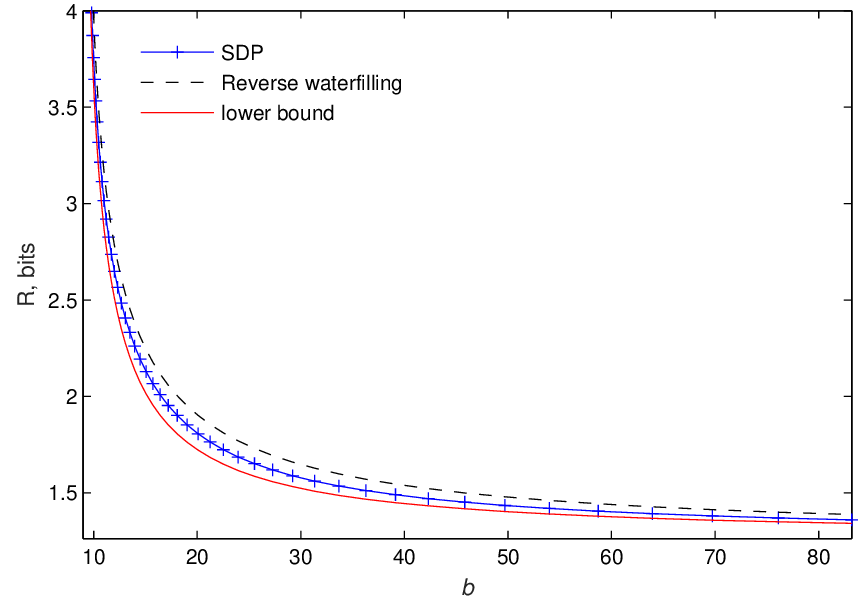,width=1\linewidth}
\end{center}
 \caption{Our lower bound in \thmref{thm:mainu}, the exact rate-cost function computed using the SDP method \cite{tanaka2017lqg}, and the reverse watefilling solution \cite{tatikonda2004stochastic}, computed for a Gaussian $V$,  $\mathsf A = \diag(2, 1.2, .3)$, $\mathsf Q = \mathsf R = \mathsf I$; $\mathsf B$ and $\mathsf \Sigma_{V}$ were generated randomly. The reverse watefilling solution is evidently strictly suboptimal.} 
 \label{fig:rdl}
\end{figure}

\subsection{Partially observed system}
\label{sec:po}
Consider now the scenario in which the encoder sees a noisy observation of the system state and forms a codeword to transmit to the controller using its present and past noisy observations. If the system and observation noises are jointly Gaussian, our results in \secref{sec:fo} generalize readily. 

In the absence of communication constraints, the minimum cost decomposes into two terms, the cost due to the noise in the system in \eqref{eq:cmin}, and the cost due to the noise in the observation of the state:
 \begin{equation}
 b_{\min} = \tr(\mathsf \Sigma_V \mathsf S) +  \tr\left( \mathsf \Sigma \mathsf A^T \mathsf M \mathsf A\right), \label{eq:bdecomp}
\end{equation}
where $\mat \Sigma$ is the covariance matrix of the estimation error $X_t - \E{X_t | Y^t, U^{t-1}}$ in the limit $t \to \infty$. The celebrated {\it separation principle} of estimation and control states that the minimum cost in \eqref{eq:bdecomp} can be attained by separately estimating the value of $X_i$ using the noisy observations  $Y^i $, and by applying the optimal control to the estimate as if it was the true state of the system. %
If system and observation noises are Gaussian, the optimal estimator admits a particularly elegant implementation via the Kalman filter. Then, $\mat \Sigma$ is given by
\begin{equation}
\mat \Sigma = %
\mat P - \mathsf K \left( \mathsf C \mat P \mathsf C^T + \mathsf \Sigma_W \right) \mathsf K^T,
\end{equation}
where $\mathsf \Sigma_W$ is the covariance matrix of each $W_i$, $\mat P$ is the solution to the algebraic Riccati equation
\begin{align}
\mat P &= %
\mathsf A \mat P  \mathsf A^T - \mathsf A \mathsf K \left( \mathsf C \mat P \mathsf C^T + \mathsf \Sigma_W \right) \mathsf K^T \mathsf A^T + \mathsf \Sigma_V, \label{eq:Priccati}
\end{align}
and $\mathsf K$ is the steady state Kalman filter gain:
\begin{equation}
 \mathsf K \triangleq \mat P \mathsf C^T \left( \mathsf C \mat P \mathsf C^T + \mathsf \Sigma_W \right)^{-1}. \label{eq:kalmangain}
\end{equation}

As Theorem \ref{thm:mainpo} below shows, the rate-cost function is bounded in terms of the steady state covariance matrix of the innovation in encoder's state estimate: 
\begin{align}
\mat N &\triangleq  \mathsf K \left( \mathsf C  \mat P  \mathsf C^T + \mathsf \Sigma_W \right) \mat K^T%
= \mathsf A \mat \Sigma \mathsf A^T - \mat \Sigma + \mathsf \Sigma_V.
\end{align}

\begin{thm}
 \label{thm:mainpo}
Consider the partially observed linear stochastic system \eqref{eq:system1}, \eqref{eq:system2}. Suppose further that $X_1, V$ and $W$ are all Gaussian.  At any target LQR cost $b > b_{\min}$, the rate-cost function is bounded below as
\begin{align}
\mathbb R(b) \geq \log |\det \mathsf A | + \frac n 2 \log \left( 1 + \frac{ \left( \det \mat N \mat M \right)^{\frac 1 n}} {(b - b_{\min})/n }\right). \label{eq:mainpog}
\end{align}
\end{thm}
The bound in \thmref{thm:mainpo} is nontrivial if both $\mat M \succ \mat 0$ and $\mat N \succ \mat 0$. The necessary condition for that is $\rank \mathsf B = \rank \mat C = n$. 
In the fully observed system, $\mathsf C = \mathsf K = \mathsf I$, $\mat P = \mathsf N = \mathsf \Sigma_{ V}$, $\mathsf \Sigma_{ W} = \mathsf 0$, and \eqref{eq:mainpog} reduces to \eqref{eq:maing}. %
As we will see in \secref{sec:c}, the key to proving \thmref{thm:mainpo} is to represent the evolution of the current encoder's estimate of the state in terms of its previous best estimate and an independent Gaussian, as carried out by the Kalman filter recursion. It is for this technique to apply that we require the noises to be Gaussian.\footnote{Recall that Theorems \ref{thm:main}, \ref{thm:maina}, \ref{thm:mainu} make no such restrictions.} \vlong{\figref{fig:rdpo} displays the lower bound in \thmref{thm:mainpo} as a function of the target cost $b$, together with the entropy of a simple uniform variable-rate quantizer.}{} As for the achievable scheme, as in the fully observed case, the observer quantizes the innovation, only now the innovation is the difference between the observer's estimate and the controller's estimate. The scheme operates as follows. The observer sees $Y_i$, recursively computes its estimate of the state $X_i$, computes the innovation, quantizes the innovation. The controller receives the quantized value of the innovation, recursively computes its own estimate of the state $X_i$, and forms the control action that is a function of controller's state estimate only. Accordingly, the tasks of state estimation, quantization and control are separated. The momentous insight afforded by Theorem \ref{thm:mainpoa} below is that in the high rate regime, this simple scheme performs provably close to $\mathbb R(b)$, the best rate-cost tradeoff theoretically attainable. 

\vlong{
\begin{figure}[htp]
\begin{center}
    \epsfig{file=rdpo.eps,width=1\linewidth}
\end{center}
 \caption[]{The minimum quantizer entropy compatible with cost $b$ in a partially observed LQG system \eqref{eq:system1}, \eqref{eq:system2} with parameters $n = 1$, $\mathsf A = 2$, $\mathsf B = \mathsf Q = \mathsf R = 1$, $V \sim \mathcal N(0, 1)$, $W \sim \mathcal N(0, 1)$.} \label{fig:rdpo}
\end{figure}
}{}

\begin{thm}
\label{thm:mainpoa}
Consider the partially observed linear stochastic system \eqref{eq:system1}, \eqref{eq:system2}.
Suppose that $\mat M \succ 0$, $\mat N \succ 0$, and that $X_1, V$ and $W$ are all Gaussian.
At any LQR cost $b > b_{\min}$, the entropy-cost function is bounded above as
\begin{align}
\mathbb H(b) &~\leq \log |\det \mathsf A | + \frac n 2 \log \left( 1 + \frac{ \left|\det \mathsf N \mathsf M \right|^{\frac 1 n}} {(b - b_{\min})/n }\right) \notag\\
&~+ O_1(\log n) + O_2 \left( \left(b - b_{\min}\right)^{\frac 1 2}\right),\label{eq:mainapog}
\end{align}
where $O_1(\log n) \leq C_1 \log n$ and $O_2(\xi) \leq C_2 \min \left\{ \xi, c_2 \right\}$ for some nonnegative constants $C_1,\ C_2$ and $c_2$. %
\end{thm}

Theorems \ref{thm:mainu} and \ref{thm:mainm} generalize as follows. 

\begin{thm}
\label{thm:mainpou}
Consider the partially observed linear stochastic system \eqref{eq:system1}, \eqref{eq:system2}. Suppose that $(\mathsf A, \mathsf B)$ is controllable and $(\mathsf A, \mathsf C)$ is observable.  Suppose further that $X_1, V$ and $W$ are all Gaussian. At any LQR cost $b >  b_{\min}$, the rate-cost function is bounded from below as
\begin{align}
 \mathbb R(b) \geq \log a^\prime + \frac 1 2 \log \left( 1 + \frac{ \eta^\prime \mu^\prime}{(b - b_{\min})/{\ell} }\right).\label{eq:mainpogen} 
\end{align}
where $a^\prime$, $\mu^\prime$ are defined in \eqref{eq:aprime}-\eqref{eq:muprime}, and 
\begin{align}
\eta^\prime \triangleq  \left( \det \left(\mat \Pi_\ell^T \mat J^{-1} \mat N \mat J^{-1\, T} \mat \Pi_\ell\right) \right)^{\frac 1 \ell},
\end{align}
where $\mat J$ and $\mat \Pi_{\ell}$ are defined in \eqref{eq:Aprime}, \eqref{eq:pil}, respectively.
\end{thm}

\begin{thm}
\label{thm:mainpom}
Consider the partially observed linear stochastic system \eqref{eq:system1}, \eqref{eq:system2}. Suppose that $(\mathsf A, \mathsf B)$ is controllable, $(\mathsf A, \mathsf C)$ is observable, and that $m \leq k \leq n$. 
Suppose further that $X_1, V$ and $W$ are all Gaussian. 
 At any LQR cost $b >  b_{\min}$, the rate-cost function is bounded from below as
\begin{flalign}
&\mathbb R(b) \geq \log |\det \mat A| \!+\! \frac m 2 \log \left( \frac{a^2}{|\det \mat A|^2}\! +\! \frac{ \eta \mu}{( b \!-\!  b_{\min} )/m}\right),\hspace{-4em}  &\label{eq:rbpom}
\end{flalign}
where $\mu$ is defined in \eqref{eq:mugen}, 
\begin{align}
a &\triangleq \inf_{i \geq 1} \left(\frac{ \det \left( \mathsf L \mathsf A^{i} \mathsf K \left( \mathsf C  \mat P  \mathsf C^T + \mathsf \Sigma_W \right) \mat K^T \mathsf A^{i\, T} \mathsf L^T \right)}{\det \left( \mathsf C  \mat P  \mathsf C^T + \mathsf \Sigma_W \right)  \det \left(\mat L \mat L^T\right)  \det \left(\mat K^T \mat K \right)  } \right)^{\frac 1 {2 i m}},\notag\\
\eta &\triangleq  \left(\det \left( \mathsf C  \mat P  \mathsf C^T + \mathsf \Sigma_W \right) \det\left(  \mat K^T \mathsf K \right)\right)^{\frac 1 m},
\end{align}
and $\mat L$, $\mat P$, $\mat K$ are defined in \eqref{eq:L}, \eqref{eq:Priccati}, \eqref{eq:kalmangain}, respectively.
\end{thm}
If $k = m = n$,  \thmref{thm:mainpom} reduces to \thmref{thm:mainpo}.

\subsection{Operational implications}
\label{sec:oper}
So far we have formally defined the rate-cost / entropy-cost functions as the limiting solutions of a minimal directed mutual information / entropy subject to an LQR cost constraint (\defnref{defn:rb} / \defnref{defn:hb}) and presented lower and upper bounds to those functions. In this section we discuss the operational implications of the results in \secref{sec:fo} and \secref{sec:po} in several communication scenarios.

\subsubsection{Control over a noisy channel}
Consider first a general scenario in which the channel in Fig. \ref{fig:system} is a dynamical channel defined by causal kernels $\{P_{G^t \| F^t}\}_{t = 1}^\infty$. %
For the class of  \emph{directed information stable} channels, the feedback capacity of the channel is \cite{kim2008coding,tatikonda2009feedbackcapacity}
\begin{equation}
C = \liminf_{t \to \infty}  \sup_{P_{F^t \| \mathcal D G^t}} \frac 1 tI(F^t \to G^t),  
\end{equation}
If past channel outputs are not available at the encoder, then the $\sup$ is over all $P_{F^t}$, and $I(F^t \to G^t) = I(F^t; G^t)$.

The following result, the proof of which is deferred until \secref{sec:sep}, implies that the converse results in Theorems~\ref{thm:main}, \ref{thm:mainu}, \ref{thm:mainpo} and \ref{thm:mainpou} present lower bounds on the capacity of the channel necessary to attain $b$. 
\begin{prop}
\label{prop:dpconverse}
A necessary condition for stabilizing the  system in \eqref{eq:system1}, \eqref{eq:system2} at LQR cost $b$ is, 
\begin{equation}
\mathbb R(b) \leq C \label{eq:ClRb}.
\end{equation}
\end{prop}
One remarkable special case when equality in \eqref{eq:ClRb} is attained is control of a scalar Gaussian system over a scalar memoryless AWGN channel \cite{bansal1989simultaneous,tatikonda2004stochastic,freudenberg2010stabilization,khina2016multi}. In that case, the channel is probabilistically matched to the data to be transmitted \cite{gastpar2003tocodeornot}, no coding beyond simple amplification is needed, and linearly transmitting the innovation is optimal \cite{bansal1989simultaneous}. In practice, such matching rarely occurs, and intelligent joint source-channel coding techniques can lead to a significant performance improvement.  One such technique that approaches \eqref{eq:ClRb} in a particular scenario is discussed in \cite{khina2016multi}. In general, how closely the bound in \eqref{eq:ClRb} can be approached over noisy channels remains an open problem.

\subsubsection{Control under fixed-rate quantization}
\label{sec:fixed}
If the channel connecting the encoder to the controller is a noiseless bit pipe that accepts a fixed number of $r$ bits per channel use (so that $C = r$), both $\mathbb R(b)$ and $\mathbb H(b)$  are lower bounds on the minimum quantization rate $r$ required to attain cost $b$:
\begin{equation}
 \mathbb R(b) \leq \mathbb H(b) \leq r. 
\end{equation}
Therefore, the converse results in Theorems \ref{thm:main}, \ref{thm:mainu}, \ref{thm:mainm}, \ref{thm:mainpo}, \ref{thm:mainpou} and \ref{thm:mainpom} give sharp lower bounds on the minimum size of a fixed-rate quantizer compatible with LQR cost $b$. The achievability results in Theorems \ref{thm:maina} and \ref{thm:mainpoa} are insufficient to establish the existence of a fixed-rate quantizer of a rate approaching $\mathbb R(b)$. While attempting to find a time-invariant fixed-rate quantizer operating at any finite cost is futile \cite{nair2004stabilizability}, determining whether there exists an adaptive fixed-rate quantization scheme approaching the converses in Theorems~\ref{thm:main} and \ref{thm:mainpo} remains an intriguing open question.

\subsubsection{Control under variable-length quantization}
\label{sec:avrate}
Assume now that the channel connecting the encoder to the controller is a noiseless channel that accepts an average of $r$ bits per sample. That is, the channel input alphabet is  the set of all binary strings, 
 $\{\emptyset, 0, 1, 00, 01, \ldots\}$,
and the encoding function $P_{F^t \| Y^t}$ must be such that 
\begin{equation}
\frac 1 t \sum_{i = 1}^{t} \E{\ell(F_i)} \leq r, \label{eq:varq}
\end{equation}
where $\ell(F_i)$ denotes the length of the binary string $F_i$.

The minimum encoded average length $L^\star(X)$ in lossless compression of object $X$ is bounded as \cite{wyner1972upper,alon1994lower}
\begin{align}
L^\star(X) &\leq H(X) \label{eq:wyner}\\
&\leq		L^\star(X) 
		+ \log_2 (L^\star(X)+1) + \log_2 e  \label{eq:alon2}.
\end{align}	
Note that \eqref{eq:wyner} states that the optimum compressed length is below the entropy. This is a consequence of lifting the prefix condition: without prefix constraints one can compress at an average rate slightly below the entropy \cite{wyner1972upper,szpankowski2011minimum}. 

The operational rate-cost function for control under variable-length quantization, $R_{\mathrm{var}}(b)$, is defined as the limsup of the infimum of $r$'s such that LQR cost $b$ and \eqref{eq:varq} are achievable in the limit of large $t$. It follows from \eqref{eq:wyner} that
\begin{align}
R_{\mathrm{var}}(b) &\leq \mathbb H(b) \label{eq:Rvara},%
\end{align}
which implies the existence of a variable-length quantizer whose average rate does not exceed the expressions in Theorems \ref{thm:maina} and \ref{thm:mainpoa}. Likewise, by \eqref{eq:alon2} and Jensen's inequality,
\begin{align}
\mathbb H(b) %
&\leq  R_{\mathrm{var}}(b) + \log_2 ( R_{\mathrm{var}}(b)  + 1) + \log_2 e, \label{eq:Rvarc}
\end{align}
which leads via \eqref{eq:hlr} to the lower bound $R_{\mathrm{var}}(b) \geq \psi^{-1}(\mathbb R(b))$, where $\psi^{-1}(\cdot)$ is the functional inverse of $\psi(x) = x + \log_2 (x + 1) + \log_2 e$. Thus, Theorems \ref{thm:main}, \ref{thm:mainu}, \ref{thm:mainm}, \ref{thm:mainpo}, \ref{thm:mainpou} and \ref{thm:mainpom} provide lower bounds to $R_{\mathrm{var}}(b)$. Consequently, our converse and achievability results in \secref{sec:fo} and \secref{sec:po} characterize the operational rate-cost tradeoff for control with variable-length quantizers.

\begin{remark}
The minimum average compressed length among all prefix-free lossless compressors, $L_{\text{p}}^\star(X)$, satisfies  
$
 H(X)  \leq L_{\mathrm{p}}^\star(X) \leq H(X) + 1.
$
Therefore, the minimum average rate of a prefix-free quantizer compatible with cost $b$ is bounded as
$
 \mathbb R(b) \leq R_{\mathrm{var, p}}(b) \leq \mathbb H(b) + 1.
$
Accordingly, the theorems in \secref{sec:fo} and \secref{sec:po} also characterize the operational rate-cost tradeoff for control with variable-length prefix-free quantizers. 
\end{remark}

\vspace{-.4em}
\section{Causal Shannon lower bound}
\label{sec:causalRd}

In this section, we consider causal compression of a discrete-time random process, $S^\infty = (S_1, S_2, \ldots)$, under the weighted mean-square error (MSE): 
\begin{flalign}
 &\wmse{S^{t}, \hat S^{t}}  \label{eq:wmse}
\triangleq  \sum_{i = 1}^{t} \E{(S_i - \hat S_i)^T \mat W_i (S_i - \hat S_i) }, \hspace{-1em}&
\end{flalign}
where $\mat W_i \geq 0$.
Understanding causal compression bears great independent interest and, due to separation between quantization and control, it is also vital to understanding quantized control, as explained in \secref{sec:sep} below.

Causally conditioned directed information and causally conditioned entropy are defined as:
\begin{align}
I(S^t \to \hat S^t \| Z^t) &\triangleq \sum_{i = 1}^{t} I(S^i; \hat S_i | \hat S^{i - 1}, Z^i ) \label{eq:Idircond},\\
H(\hat S^t \| Z^t) &\triangleq \sum_{i= 1}^{t} H(\hat S_i | \hat S^{i-1}, Z^i). \label{eq:Hdircond}
\end{align}

\begin{defn}[causal rate- and entropy-distortion functions]
\label{defn:Rc}
The causal rate- and entropy-distortion functions under the weighted MSE in \eqref{eq:wmse} with side information $Z^\infty$ causally available at both encoder and decoder  are defined as, 
\begin{flalign}
&\mathbb R_{S^\infty \| Z^\infty} (d) \triangleq \limsup_{t \to \infty} \!\!\!\!\!\! \inf_{\substack{ P_{\hat S^t \| S^t, Z^{t}} \colon \\
\frac 1 t \wmse{S^{t},\, \hat S^{t}} \leq d}} \!\!\!
  \frac 1 t I(S^t \!\! \to \! \hat S^t \| Z^t), \hspace{-2em} &&   \label{eq:causalRd}\\
&\mathbb H_{S^\infty \| Z^\infty} (d) \triangleq \limsup_{t \to \infty}  \inf_{\substack{ P_{\hat S^t \| S^t, Z^{t}} \colon \\
\frac 1 t \wmse{S^{t},\, \hat S^{t}} \leq d}}
  \frac 1 t H(\hat S^t \| Z^t). \hspace{-2em} && 
\end{flalign}
\end{defn}
In the absence of side information, we write $\mathbb R_{S^\infty}(d)$ / $\mathbb H_{S^\infty}(d)$ for the causal rate- / entropy-distortion functions.

{\it DPCM encoder}, upon observing the current source sample $S_i$, computes the state innovation $\tilde S_i$ recursively using 
\begin{equation}
\tilde S_i \triangleq S_i - \hat S_{i | i-1}, \label{eq:innoDPCM}
\end{equation}
where $\hat S_{i | i-1} \triangleq \E{S_i | \hat S^{i - 1}}$ is the a priori (predicted) state estimate at the decoder given previous decoder's outputs $\hat S^{i-1}$. The DPCM encoder sends quantized innovation $\hat{\tilde S}_{i}$.

{\it DPCM decoder}, having recovered $\hat{\tilde S}_{i}$, forms its estimate:
\begin{align}
 \hat S_{i} &= \hat S_{i | i-1} + \hat{\tilde S}_{i}. 
\end{align}

\propref{prop:dpcm} below implies that the causal rate- and entropy-distortion functions are attained in the class of DPCM (let $T_i = - \hat S_{i | i-1}$ in \propref{prop:dpcm}).
No independence among samples of either the innovation process $\{\tilde S_i\}$ or its encoded version $\{\hat{\tilde S}_{i}\}$ is required for this to hold.

\begin{prop}
Let the stochastic process $\{T_i\}$ be adapted to the filtration generated by $\{Z^i\}$. Then, 
\begin{align}
 \mathbb R_{S^\infty \| Z^\infty} (d) &=  \mathbb R_{ S^\infty + T^\infty \| Z^\infty} (d), \label{eq:Rcausalsideinfoadd}\\
  \mathbb H_{S^\infty \| Z^\infty} (d) &=  \mathbb H_{S^\infty + T^\infty \| Z^\infty} (d). \label{eq:Hcausalsideinfoadd}
\end{align} 
\label{prop:dpcm}
\end{prop}
\begin{proof}
Since mutual information, entropy and distortion measure \eqref{eq:wmse} are all invariant to shifts and $T_i$ is a common knowledge to both encoder and decoder at time $i$, $I(S^t \to \hat S^t \| Z^t) = I(S^t + T^t \to \hat S^t + T^t \| Z^t)$, $H(\hat S^t \| Z^t) = H(\hat S^t + T^t \| Z^t)$, $\wmse{S^{t}, \hat S^{t}} = \wmse{S^{t} + T^t, \hat S^{t} + T^t}$, and \eqref{eq:Rcausalsideinfoadd}, \eqref{eq:Hcausalsideinfoadd} follow. 
\end{proof}

\propref{prop:dpcm} implies further that the rate-distortion tradeoffs for controlled and uncontrolled processes are the same. Indeed, 
consider the Markov process 
\begin{equation}
S_{i + 1} = \mat A S_i  + V_i,  \label{eq:St}
\end{equation}
obtained by letting $U_i \equiv 0$, $i = 1, 2, \ldots$ in \eqref{eq:system1}. The uncontrolled process \eqref{eq:St}  and the controlled one \eqref{eq:system1} are related through $ X_{i} = S_{i} + \sum_{j = 1}^{i-1} \mat A^{i - 1 - j} \mat B U_{j}$.
Provided that the encoder and the decoder both have access to past controls,  \propref{prop:dpcm} with $T_i = \sum_{j = 1}^{i-1} \mat A^{i - 1 - j} \mat B U_{j}$ yields
\begin{align}
 \mathbb R_{S^\infty} (d) &=  \mathbb R_{X^\infty \| \mathcal D U^\infty} (d), \label{eq:controlledratecost}\\
   \mathbb H_{S^\infty } (d) &=  \mathbb H_{X^\infty \| \mathcal D U^\infty} (d),\label{eq:controlledhcost}
\end{align}
where $\mathcal D U^{\infty}$ signifies causal availability of past controls. The choice of these controls does not affect the rate-distortion tradeoffs in \eqref{eq:controlledratecost} and \eqref{eq:controlledhcost}. 
Furthermore, using $\hat S_{i | i-1} = \mat A \hat S_{i - 1}$, $\hat X_{i | i-1} = \mat A \hat X_{i-1} + \mat B U_{i-1}$,  it is easy to show that both processes \eqref{eq:St} and \eqref{eq:system1} have the same innovation process: $\tilde X_i = \tilde S_i$. Thus the same DPCM scheme can be used to encode both.

The following bound is a major result of this paper.

\begin{thm}
\label{thm:seqslb}
For the Markov process \eqref{eq:St} and a sequence of weight matrices such that
\begin{equation}
\lim_{i \to \infty}\left( \det \mat W_i \right)^{\frac 1 n} = w > 0,  \label{eq:liminfass}
\end{equation}
causal rate-distortion function is bounded below as 
\begin{flalign}
&& \mathbb R_{S^\infty}(d)  &\geq \frac n 2 \log \left( a^2 + \frac {w N(V)}{d/n}\right), \label{eq:slbseq}\\
 &\text{where}  & a &\triangleq  |\det \mat A|^{\frac 1 n}.
\end{flalign}
\end{thm}
The proof of \thmref{thm:seqslb} uses the classical Shannon lower bound, together with a dynamic programming argument. \thmref{thm:seqslb} can thereby be viewed as an extension of Shannon's lower bound to the causal compression setting.

\begin{remark}
For scalar Gauss-Markov sources, equality holds in \eqref{eq:slbseq} as long as the expression under the $\log$ is $\geq 1$, recovering the known result \cite{gorbunov1974prognostic,tatikonda2004stochastic}, \cite[Th. 3]{derpich2012uppercausal}
\begin{equation}
 \mathbb R_{S^\infty}(d)  = \frac 1 2 \left|\log \left( a^2 + \frac {w \Var{V}}{d}\right)\right|_+. \label{eq:Rdgauss}
\end{equation}
\label{rem:gm}
\end{remark}

 \begin{remark}
In lieu of the long-term average (Ces\`aro mean) constraint in \eqref{eq:wmse}, \cite{gorbunov1974prognostic,tatikonda2004stochastic} considered a more stringent constraint in which the average distortion is bounded at each time instant $i$. In the infinite time horizon, for Gauss-Markov sources, both formulations are equivalent, as optimal rates and distortions settle down to their steady states \cite[Th. 3]{derpich2012uppercausal}, \cite{tanaka2015stationary}. For the scalar case this equivalence also follows by comparing \eqref{eq:slbseq} (obtained with a Ces\`aro mean constraint) and \eqref{eq:Rdgauss} (obtained in \cite{gorbunov1974prognostic,tatikonda2004stochastic} with a pointwise constaint).
\label{rem:steady}
\end{remark}

\thmref{thm:slba}, stated next, shows that the converse in \thmref{thm:seqslb} can be approached by a DPCM quantization scheme with uniform rate and distortion allocations (in the limit of infinite time horizon). This implies that  nonuniform rate and distortion allocations permitted by \eqref{eq:wmse} and \eqref{eq:Idircond}--\eqref{eq:Hdircond} cannot yield significant performance gains.  
\begin{thm}
\label{thm:slba} 
For the Markov process \eqref{eq:St} such that $V$ has a regular density and a sequence of weight matrices such that \eqref{eq:liminfass} holds,  
the causal entropy function is bounded by
\begin{flalign}
& \mathbb H_{S^\infty}(d) \! \leq \!\!
~ \frac {n} 2 \log \! \left(\! a^2 \!+\! \frac{w  N(V)}{d/n }\! \right) %
\!+\!%
O_1\!(\log n) \!+ \! O_2\! \left( \! d^{\frac 1 2} \!\right)\!\!, \hspace{-5em}&   \label{eq:slba} 
\end{flalign}
where $O_1(\log n) \leq C_1 \log n$ and $O_2(\xi) \leq C_2 \min \left\{ \xi, c_2 \right\}$ for some nonnegative constants $C_1,\ C_2$ and $c_2$. 
Furthermore, \eqref{eq:slba} is attained by a DPCM quantizer with output ${\hat S_i}$ such that $\lim_{i \to \infty} \E{(S_i - \hat S_i)^T \mat W_i (S_i - \hat S_i) } \leq d$ and $\lim_{t \to \infty} H(\hat S_i | \hat S^{i-1}) \leq $ the right side of \eqref{eq:slba}. 
\end{thm}

We conclude \secref{sec:causalRd} with two extensions of \thmref{thm:seqslb}.
If the dynamic range of the eigenvalues of $\mat A$ is large, the bound in \thmref{thm:seqslb+} below, obtained by projecting out a subset of smaller eigenvalues of $\mat A$, can provide an improvement over \thmref{thm:seqslb} at medium to large $d$. 
Recall the definition of matrices $\mat J$, $\mat \Pi_{\ell}$ in \eqref{eq:Aprime}, \eqref{eq:pil}, respectively. 
\begin{thm}
\label{thm:seqslb+}
Consider the uncontrolled process \eqref{eq:St} and the distortion in \eqref{eq:wmse}. 
Consider matrices $\mat V_i$ such that $0 \preceq \mat V_i^{T} \mat V_i \preceq  \mat J^{T}  \mat W_i \mat J $, and $\mat V_i$ commutes with $\mat \Pi_{\ell} \mat \Pi_{\ell}^T$. Assume that
\begin{equation}
 w^\prime \triangleq \liminf_{i \to \infty} \left( \det \left(\mathsf \Pi_{\ell}^T \mat V_i^{T} \mat V_i \mat \Pi_{\ell} \right)\right)^{\frac 1 {m}} > 0.  \label{eq:liminfass+}
\end{equation}
 The causal rate-distortion function is bounded below as,
\begin{equation}
\mathbb R_{S^\infty}(d)  \geq \frac \ell 2 \log \left( a^{\prime\,2} + \frac {w^\prime N(\mat \Pi_{\ell}^T \mat J^{-1} V)}{d/\ell}\right), \label{eq:slbseq+}
\end{equation}
where $a^\prime$ is defined in \eqref{eq:aprime}.
\end{thm}

\thmref{thm:seqslb+} implies that $\mathbb R_{S^\infty}(d)$ is bounded as in \eqref{eq:mainu}.
If  $\mat W_i$ is singular, or if $V_i$ does not have a density with respect to the Lebesgue measure on $\mathbb R^n$, the bounds in \thmref{thm:seqslb} and \ref{thm:seqslb+} reduce to $n  \log a$ and $n \log a^\prime$, respectively, losing the dependence on $d$.  The bound in \thmref{thm:seqslb++} below is a decreasing function of $d$, even if $\mat W_i$ is singular, or if $V_i$ is supported on a subspace of $\mathbb R^n$. 
\begin{thm}
\label{thm:seqslb++}
Consider the uncontrolled process \eqref{eq:St} and the distortion in \eqref{eq:wmse}. Assume that the weight matrices $\mat W_i$ satisfy 
$\lim_{i \to \infty} \mat W_i = \mat L^T \mat L$,
where $\mat L$ is an $m \times n$ matrix, $m \leq n$. Suppose further that $V_i = \mat K_i V^\prime_i$, where $V^\prime_i$ is a $k$-dimensional random vector with covariance matrix $\mat \Sigma_{V^\prime}$, where $k \geq m$, and $\mat K_i$ are $n \times k$ matrices such that $\lim_{i \to \infty} \mat K_i = \mat K$.  The causal rate-distortion function is bounded below as
\begin{equation}
\mathbb R_{S^\infty}(d)  \geq \frac m 2 \log \left( a^{2} + \frac {w N(V^\prime)^{\frac k m}}{d/m}\right), \label{eq:slbseq+}
\end{equation}
\begin{flalign}
&\text{where } & a &=  \inf_{i \geq 1}  \left(\frac{ \det \left( \mathsf L \mathsf A^{i} \mat K \mathsf \Sigma_{V^\prime} \mat K^T \mathsf A^{i\, T} \mathsf L^{T} \right)}{ \det \mathsf \Sigma_{V^\prime} \det \left(\mat L \mat L^T\right) \det \left(\mat K^T \mat K \right) } \right)^{\frac 1 {2 i m}}\hspace{-1.5em},\\
&& w &= \left( \det \left(\mat L  \mat L^T\right) \det \left(\mat K^T  \mat K \right) \right)^\frac{1}{m}.
\end{flalign}
\end{thm}
If $k = m = n$, \thmref{thm:seqslb++} reduces to \thmref{thm:seqslb}.

\section{C\hspace{-.6pt}o\hspace{-.6pt}n\hspace{-.6pt}t\hspace{-.6pt}r\hspace{-.6pt}o\hspace{-.6pt}l, 
e\hspace{-.6pt}s\hspace{-.6pt}t\hspace{-.6pt}i\hspace{-.6pt}m\hspace{-.6pt}a\hspace{-.6pt}t\hspace{-.6pt}i\hspace{-.6pt}o\hspace{-.6pt}n
a\hspace{-.6pt}n\hspace{-.6pt}d
c\hspace{-.6pt}o\hspace{-.6pt}m\hspace{-.6pt}m\hspace{-.6pt}u\hspace{-.6pt}n\hspace{-.6pt}i\hspace{-.6pt}c\hspace{-.6pt}a\hspace{-.6pt}t\hspace{-.6pt}i\hspace{-.6pt}o\hspace{-.6pt}n 
s\hspace{-.6pt}e\hspace{-.6pt}p\hspace{-.6pt}a\hspace{-.6pt}r\hspace{-.6pt}a\hspace{-.6pt}t\hspace{-.6pt}e\hspace{-.6pt}d}
\label{sec:sep}

An early quantization-control separation result for Gaussian systems was proposed by Fischer~\cite{fischer1982optimalquantizedcontrol}. Tatikonda~et~al.~\cite{tatikonda2004stochastic} considered control of fully observed system over a noisy channel and showed that certainty equivalence control is optimal if and only if control has no \emph{dual effect}, that is to say the present control cannot affect the future state uncertainty. Here, we observe that as long as both the encoder and the controller have access to past controls, separated design for control over noisy channels is optimal, both for fully observed systems and for Gaussian partially observed systems.

Recall the last-step weight matrix $\mathsf S_{t}$ in \eqref{eq:lqr}, and let
$\mathsf S_i$, $1 \leq i \leq t-1$ be the solution to the Riccati recursion, 
\begin{align}
\mathsf S_i = \mathsf Q + \mathsf A^T \left( \mathsf S_{i+1} -  \mathsf M_i \right) \mathsf A.
\end{align} 
\begin{flalign}
&\text{where } & \mathsf M_i &\triangleq  \mathsf S_{i+1} \mathsf B  ( \mathsf R + \mathsf B^T \mathsf S_{i+1} \mathsf B)^{-1}  \mathsf B^T \mathsf S_{i+1}\\
&& &= \mat L_i^T ( \mathsf R + \mathsf B^T \mathsf S_{i+1} \mathsf B) \mat L_i,\\
&& \mathsf L_i &\triangleq 
( \mathsf R + \mathsf B^T \mathsf S_{i+1} \mathsf B)^{-1} \mathsf B^T \mathsf S_{i + 1}.
\end{flalign}

The optimal control of partially observed system in \eqref{eq:system1}, \eqref{eq:system2} in the scenario of \figref{fig:system}, for a fixed dynamical channel $\{P_{G_i | G^{i-1}, F^i}\}_{i = 1}^\infty$ can be obtained as follows: 
\begin{enumerate}[(i)]
\item \label{i} The encoder computes the optimal state estimate
\begin{equation}
\hat X_i \triangleq \E{X_i | Y^i, U^{i-1}}. \label{eq:estt}
\end{equation}
\item The encoder maps the optimal control signal $U_i^\star = - \mat L_i \mat A \hat X_i$ to channel codeword $F_i$; 
\item Having received the channel output $G_i$, the decoder computes $U_i$ and applies it to the system. 
\label{iii}
\end{enumerate}
The encoder $\{P_{F_i | Y^i, U^{i-1}}\}$ and the controller $\{P_{U_i | G^{i}, U^{i-1}}\}$ are designed to minimize the LQR cost, written as follows. 
\begin{thm}
\label{thm:sep}
The LQR cost in the scenario of \figref{fig:system} for the partially observed system \eqref{eq:system1}, \eqref{eq:system2} separates as, 
\begin{equation}
\lqr{X^{t}, U^{t-1}} = \sum_{i = 0}^{t-1} c_{i} +  \sum_{i = 1}^{t-1} e_{i} +  \sum_{i = 1}^{t-1} d_{i}, \label{eq:sep}
\end{equation}
where the control, estimation and communication costs are respectively given by 
\begin{align}
c_i &\triangleq %
\tr\left(\mathsf \Sigma_V \mathsf S_{i+1} \right), \quad c_0 \triangleq \tr\left(\mathsf \Sigma_{X_1} \mathsf S_{1} \right) \label{eq:ci}\\
e_i &\triangleq \E{
 \left( X_i - \hat X_i  \right)^T  
 \mat A^T  \mat M_i \mat A
 \left( X_i - \hat X_i  \right) },\\
 d_i &\triangleq \E{(U_i - U_i^\star)^T ( \mathsf R + \mathsf B^T \mathsf S_{i+1} \mathsf B) (U_i - U_i^\star)}. \label{eq:di}
 \end{align} 
\end{thm}

\begin{proof}
Denote for $1 \leq i \leq t$
\begin{align}
\!\!\! b_{i} &\triangleq 
\E{ \sum_{j = 1}^{i-1} \left(\vect X_{j}^T \mathsf Q X_{j}  + \vect U_j^T \mathsf R U_j \right)} +  \E{X_{i}^T \mathsf S_{i} X_{i}} \label{eq:bi}.
\end{align}
Eliminating $X_{i}$ from  \eqref{eq:bi} by substituting $X_{i}= \mat A X_{i-1} + \mat B U_{i-1} + V_{i-1}$ into \eqref{eq:bi}, completing the squares and using that $V_{i-1}$ is zero-mean and independent of $X_{i-1}, U_{i-1}$, we re-write \eqref{eq:bi} equivalently as
\begin{align}
b_{i} &=  b_{i-1} + \E{V_{i-1}^T \mathsf S_{i} V_{i-1}} + q_{i-1} \label{eq:bdrecursive}, 
\end{align}
where
\begin{align}
\!\!\!\!\! q_i &\triangleq  \E{
 \left( \mathsf L_i \mat A X_i + U_i \right)^T\!\!
( \mathsf R + \mathsf B^T \mathsf S_{i+1} \mathsf B)\!
 \left( \mathsf L_i \mat A X_i + U_i \right)}\!\!. \hspace{-10pt} \label{eq:qi}
\end{align}
Applying \eqref{eq:bdrecursive} repeatedly, we obtain
\begin{align}
 b_{t} &= b_1 + \sum_{i = 1}^{t-1} \E{V_i^T \mathsf S_{i+1} V_i} + \sum_{i = 1}^{t-1} q_i, \label{eq:bsum}
\end{align}
which is equivalent to \eqref{eq:sep} in the fully observed scenario, i.e. when $\hat X_i = X_i$.

To show the more general case, observe first that for random vectors $X$, $Y$ and $\hat X$ forming a Markov chain $X - Y - \hat X $, 
it holds with $\hat X^\star \triangleq \E{X|Y}$ that,
\begin{align}
\!\!\! \E{ \| X - \hat X\|^2} &= \E{ \| X - \hat X^\star \|^2}%
 + \E{\|\hat X^\star - \hat X \|^2 }\!\!.\!\!\!  \label{eq:msedecomp}
\end{align}
Applying \eqref{eq:msedecomp} to $\mat L_i \mat A X_i - (Y^i, U^{i-1}) - U_i$, we separate $q_i$ into  
$
q_i =   e_i + d_i
$,
and thereby rewrite \eqref{eq:bsum} as \eqref{eq:sep}. 
\end{proof}
Using \thmref{thm:sep}, we write the minimum achievable LQR cost as
\begin{align}
 &~\textstyle{\sum_{i = 0}^{t-1}} c_{i} +  \inf \left\{ \textstyle{\sum_{i = 1}^{t-1}} e_{i} +   \textstyle{\sum_{i = 1}^{t-1}} d_{i} \right\} \notag\\
\geq&~  \textstyle{\sum_{i = 0}^{t-1}} c_{i} + \inf \left\{ \textstyle{\sum_{i = 1}^{t-1}} e_{i} \right\} + \inf \left\{\textstyle{\sum_{i = 1}^{t-1}} d_{i}\right\}, \label{minLQRbound} 
\end{align}
where the infimum is over all admissible control sequences. Equality in \eqref{minLQRbound} is not attained in general; two important scenarios when is it achieved are fully observed systems and for Gaussian partially observed systems. 
In the fully observed case, the estimation terms $e_i$ disappear, and in the Gaussian partially observed case, it is well known that those terms do not depend on the choice of controls. 
Indeed, if $V_i$ and $W_i$ are both Gaussian, then the optimal estimator \eqref{eq:estt} can be implemented via a linear recursion given by the Kalman filter.  At time $i$, having observed $Y_i$, the Kalman filter forms an estimate of the system state using $Y_i$ and the prior estimate $\hat X_{i-1}$ as follows  (e.g. \cite{bertsekas1995dynamic}): 
\begin{flalign}
& & \hat X_{i} &= \mathsf A \hat X_{i-1} + \mathsf B U_{i-1} + \mathsf K_{i} \tilde Y_i \label{eq:kalman},\\
&\text{where } & \tilde Y_i &\triangleq   Y_i - \mathsf C \mathsf A \hat X_{i-1} - \mathsf C \mathsf B U_{i-1} &&
\end{flalign} y
is the {\it innovation}, Gaussian and  independent of $\hat X_{i-1}$, $U_{i-1}$ and $\tilde Y^{i-1}$, 
and the Kalman filter gain is found from the Riccati recursion:
\begin{align}
\!\!\!\!\mathsf K_i &\triangleq \mathsf P_{i | i-1} \mathsf C^T \left( \mathsf C \mathsf P_{i|i-1} \mathsf C^T + \mathsf \Sigma_W \right)^{-1}, \\
\!\!\!\! \mathsf P_{i+1 | i} &= \mathsf A \left( \mathsf I - \mathsf K_i \mathsf C \right)  \mathsf P_{i | i-1}  \mathsf A^T + \mathsf \Sigma_V,
~~ \mathsf P_{1|0} \triangleq \mathsf \Sigma_{X_1}. \label{eq:Pt+1|t}
\end{align}
The covariances of the innovation and the estimation error are given by, respectively, 
\begin{align}
\Cov (\tilde Y_i) &= \mathsf C \mathsf P_{i | i-1} \mathsf C^T + \mathsf \Sigma_{W}, \label{eq:covinno}\\
\Cov( X_i - \hat X_i ) &=  (\mathsf I - \mathsf K_i \mathsf C) \mathsf P_{i|i-1}  \label{eq:coverrx},
\end{align}
and the estimation error is thus given by
\begin{equation}
e_i = \tr\left( \left(\mathsf I - \mathsf K_i \mathsf C) \mathsf P_{i|i-1}\right) \mathsf A^T \mathsf M_i \mathsf A\right). 
\end{equation}

An immediate corollary to \thmref{thm:sep} is the following. 
\begin{cor}
\label{cor:link}
In both fully observed systems and Gaussian partially observed systems, the rate-cost and the entropy-cost functions and independent of the control sequence $U^\infty$ and are given by, respectively,
\begin{align}
\mathbb R(b) &= \mathbb R_{\hat X^\infty \| \mathcal D U^{\infty}} (b - b_{\min}), \label{eq:rcausalrcp}\\
\mathbb H(b) &= \mathbb H_{\hat X^\infty \| \mathcal D U^{\infty}} (b - b_{\min}), \label{eq:rcausalrh}
\end{align}
 where $b_{\min}$ is the minimum cost attainable without communication constraints in \eqref{eq:bdecomp}, and causal rate- and  entropy-distortion functions are evaluated with weight matrices 
\begin{equation}
\mat W_i = \mat A^T \mat M_i \mat A. \label{eq:Wimat}
\end{equation}
\end{cor}
\begin{proof}
 Since equality in \eqref{minLQRbound} holds, we need to minimize $\sum_{i = 1}^{t-1} d_{i}$ subject to either directed information or entropy constraint. Once we argue that the minimal achievable distortions can be equivalently written as
\begin{equation}
 d_i = \E{( \hat X_i - \doublehat {X_i})^T \mat A^T  \mat M_i \mat A ( \hat X_i - \doublehat {X_i} )}, \label{eq:dialt}
\end{equation}
where $\doublehat {X_i}$ is the controller's estimate of $\hat X_i$, and $U_i = - \mat L_i \mat A \doublehat {X_i}$, we will immediately obtain \eqref{eq:rcausalrcp}, \eqref{eq:rcausalrh}.  But this follows via the same arguments as in the proof of \propref{prop:DrM} below, using data processing for directed information \cite[Lemma 4.8.1]{tatikonda2000thesis} in lieu of that for mutual information. 
\end{proof}

Via \corref{cor:link}, we can show the converse for control over noisy channels in \propref{prop:dpconverse} using a converse for tracking over noisy channels. 
Tracking $S_1, S_2, \ldots$ over a causal feedback channel $P_{G^t \| F^t}$ gives rise to a joint distribution of the form $P_{S^t} P_{F^t \| S^t, \mathcal D G^t} P_{G^t \| F^t} P_{\hat S^t \| G^t}$, where $P_{F^t \| S^t, \mathcal D G^t}$ and $P_{\hat S^t \| G^t}$ represent encoder and decoder mappings, and the goal is to minimize the distortion between $S^t$ and $\hat S^t$.  A necessary condition for the existence of an encoder/decoder pair achieving distortion $d$ in the limit of infinite time horizon is \cite[Th. 5.3.2]{tatikonda2000thesis}, 
\begin{equation}
\mathbb R_{S^\infty}(d) \leq C \label{eq:ClRddir}.
\end{equation}

\propref{prop:dpconverse} follows by plugging \eqref{eq:controlledratecost} and \eqref{eq:rcausalrcp} in \eqref{eq:ClRddir}.

\section{Converse theorems: tools and proofs}
\label{sec:c}

We start by introducing a few definitions and tools, some classical, some novel, that form the basis of our technique.
 
Conditional entropy power is defined as
\begin{equation}
 N(X|U) \triangleq \frac 1 {2 \pi e} \exp\left( \frac 2 n h(X|U) \right),
\end{equation}
where $h(X|U) = -\E{ \int_{\mathbb R^n} f_{X|U}(x| U) \log f_{X|U}(x | U) dx}$ is the conditional differential entropy of $X$. 

\begin{prop}
For $X \in \mathbb R^n$, 
 \begin{align}
 n N(X|U)  &\leq \Var{X|U}, \label{eq:epgauss} 
\end{align}
with equality if and only if  $X = U + S$, where $S$ is Gaussian.  
 \label{prop:gauss}
\end{prop}
\begin{proof}
The unconditional case is a well-known maximum entropy result (e.g. \cite[Example 12.2.8]{cover2012elements}). This implies that for each realization of $u$, 
$ n N(X|U = u)  \leq \Var{X|U = u}$. Taking expectation with respect to $U$ of both sides and using strict convexity of $x \mapsto \exp(x)$, we obtain \eqref{eq:epgauss} together with condition for equality. 
\end{proof}

An essential component of our analysis, the conditional entropy power inequality (EPI), follows from the unconditional EPI \cite{shannon1948mathematical,stam1959some}  using convexity of the function $(x, y) \mapsto \log \left(\exp(x) + \exp(y) \right)$ . 
\begin{thm}[Conditional EPI]
If $X \pperp Y$ given $U$, then
 \begin{align}
N(X + Y | U)\geq N(X|U) + N(Y|U) \label{eq:epi}.
\end{align}
\label{thm:epi}
\end{thm}
\vspace{-1.5em}
In causal data compression, the quantized data at current step creates the side information for the data to be compressed at the next step. The following bound to the conditional entropy power minimized over side information will be vital in proving our converse theorems.  

\begin{prop}
For $X \in \mathbb R^n$,
 \begin{align}
\inf_{P_{U|X} \colon I(X; U) \leq r} N(X | U) &\geq N(X) \exp\left( - 2r / n \right). \label{eq:unconditioning}
\end{align} 
\label{prop:unconditioning}
\end{prop}

\begin{proof}
Observe that
\begin{align}
\inf_{P_{U|X} \colon I(X; U) \leq r} h(X|U) &= \inf_{P_{U|X} \colon h(X) - h(X|U) \leq r} h(X|U) \notag\\
&\geq h(X) - r, 
\end{align}
which is equivalent to \eqref{eq:unconditioning}.  
\end{proof}

If $\mathsf L$ is square, the entropy power scales as
\begin{align}
N(\mathsf L X  | U ) = |\det \mathsf L|^\frac{2}{n} N(X | U ). \label{eq:epscale}
\end{align}
The next proposition generalizes the scaling property \eqref{eq:epscale} to the case where the multiplying matrix is not square. \footnote{\propref{prop:hlx} is stated for the unconditional case for simplicity only; naturally, its conditional version also holds. }
\begin{prop}
\label{prop:hlx}
 Let $X \in \mathbb R^n$ be a random vector with covariance $\mathsf \Sigma_X \succ 0$, let $m \leq n$, and let $\mathsf L$ be an $m \times n$ matrix with rank $m$. Then, 
\begin{align}
N( \mathsf L X) \geq  \left(\frac{ \det \left( \mathsf L \mathsf \Sigma_X \mathsf L^T \right)}{\det \mathsf \Sigma_X } \right)^{\frac 1 m} \left( N(X) \right)^{\frac n m}
 \label{eq:hlx}
\end{align}
Equality holds in \eqref{eq:hlx} if $m = n$ or if $X$ is Gaussian. 
\end{prop}

\begin{proof}
Without loss of generality, assume $\E{X} = 0$.  Express $h(X)$ through the relative entropy $D(\cdot \| \cdot )$ as
\begin{equation}
\! h(X)= \frac 1 2 \log \left( (2 \pi e)^n \det \mathsf \Sigma_X \right) - D \left( P_X \| \mathcal N(\mathbf 0, \mathsf \Sigma_X)\right)\!, \!\!\!\! \label{eq:hlx0}
\end{equation}
By the data processing inequality of relative entropy, 
\begin{align}
 &~ 
 D \left( P_X \| \mathcal N(\mathbf 0, \mathsf \Sigma_X)\right)%
 \geq 
 D \left( P_{\mathsf LX} \| \mathcal N(\mathbf 0, \mathsf L \mathsf \Sigma_X \mathsf L^T)\right) 
 \\
=&~ \frac 1 2 \log \left( (2 \pi e)^m \det \left(\mathsf L \mathsf \Sigma_X \mathsf L^T \right) \right) - h(\mathsf L X), 
 \label{eq:hlx1}
\end{align}
and \eqref{eq:hlx} follows by substituting \eqref{eq:hlx1} into \eqref{eq:hlx0} and applying \eqref{eq:ep}. 
\end{proof}

The (single-shot) distortion-rate function with respect to the weighted mean-square distortion is defined as follows. 
\begin{defn}[conditional distortion-rate function]
\label{defn:drM}
Let $X \in \mathbb R^n$ be a random vector, and $\mathsf M \succeq 0$ be an $n \times n$ matrix. The distortion-rate function under the weighted MSE with side information $U$ at both the encoder and the decoder  is 
\begin{equation}
{\mathbb D}_{r, \mathsf M}(X | U) \triangleq \!\!\!\! \inf_{\substack{ %
P_{\hat X | XU} \colon \\
I(X; \hat X|U) \leq r}} \!\!
\E{ (X - \hat X)^T \mathsf M (X -  \hat X)}\!\!.  \!\!\!\!\!\!  \label{eq:drMdefcond}
\end{equation}
\end{defn}
If no side information is available, i.e. $U \equiv 0$, we denote the corresponding unconditional distortion-rate function by ${\mathbb D}_{r, \mathsf M}(X)$.  The distortion-rate function under MSE distortion corresponds to $\mathsf M = \mathsf I$, and we simply denote 
\begin{align}
 {\mathbb D}_{r}(X | U) &\triangleq {\mathbb D}_{r, \mathsf I}(X | U).
\end{align}

The next proposition equates the distortion-rate functions under weighted and non-weighed  MSE.   
\begin{prop}
\label{prop:DrM}
 Let $X \in \mathbb R^n$ be a random vector, and let $\mat L$ be an $m \times n$ matrix. The following equality holds. 
\begin{align}
 {\mathbb D}_{r}(\mat L X | U) &=  {\mathbb D}_{r,\, \mat L^T \mat L}(X | U).  \label{eq:Ddimcond}
\end{align}
\end{prop}

\begin{proof}
We show the unconditional version of \eqref{eq:Ddimcond}; the conditional one is analogous. 
We will prove 
\begin{align}
\hspace*{-.8em} {\mathbb D}_{r}(\mat L X) %
 &\triangleq  \inf_{\hat {X}  \colon I(\mat L X; \hat {X}) \leq r}  \E{ ( \mat L X - \hat X)^T  (  \mat L X - \hat X) } \label{eq:Drldef}\\
&=\hspace*{-.3em} \inf_{\hat {X}  \colon I(\mat L X; \mat L \hat {X}) \leq r}  \hspace*{-.3em}  \E{ ( X - \hat X)^T \mat L^T \mat L  (  X - \hat X)} \label{eq:urestr}\\
&= \inf_{\hat {X}  \colon I(X; \hat X ) \leq r}  \E{ ( X - \hat X)^T \mat L^T \mat L (  X - \hat X) } \label{eq:Ddima}\\
 &\triangleq {\mathbb D}_{r, \mat L^T \mat L}(X).
 \end{align}
To show $\geq$ in \eqref{eq:urestr}, let $\mathsf \Pi$ be the orthogonal projection matrix onto the column space of $\mat L$. We use  $\| \mathsf \Pi x\| \leq \| x \|$  and $\mathsf \Pi \mat L  x  = \mat L x$ to claim $\E{ ( \mat L X - \hat X)^T  (  \mat L X - \hat X) } \geq \E{ ( \mat L X - \mathsf \Pi \hat X)^T (  \mat L X - \mathsf \Pi \hat X) }$, and data processing for mutual information to claim $I(\mat L X; \hat {X}) \geq I(\mat L X; \mat \Pi \hat {X})$. Likewise, $\leq$ holds in \eqref{eq:Ddima} by data processing. To show that $\leq$ holds in \eqref{eq:urestr}, we note that the optimization problem in  \eqref{eq:urestr} is obtained by restricting the domain of minimization in \eqref{eq:Drldef} to $\hat X \in \mathrm{Im}(\mat L)$\footnote{The image of a linear transformation described by matrix $\mat L$ is the span of its column vectors.}. To show that $\geq$ holds in \eqref{eq:Ddima}, we note that the optimization problem in \eqref{eq:urestr} is obtained by restricting the domain of minimization in \eqref{eq:Ddima} to $\hat X \in \mathrm{Im}(\mat L)$ satisfying the Markov chain condition $X - \mat L X - \hat X$, since for such $\hat X$, $I(X; \hat X ) = I(\mat L X; \hat X ) = I(\mat L X; \mat L\hat X )$.
\end{proof}

\begin{remark}
We may always assume that $X$ has uncorrelated components when computing distortion-rate functions. Indeed, let $\mat L$ be the orthogonal transformation that diagonalizes the covariance matrix of $X$. Since $\mat L^T \mat L = \mat I$, by \propref{prop:DrM} the MSE distortion-rate functions of $X$ and $\mat L X$ coincide.
\label{rem:uncor}
\end{remark}

The following tool will be instrumental in our analysis. 
\begin{thm}[Conditional Shannon lower bound]
\label{thm:slb}
The conditional distortion-rate function is bounded below as
 \begin{align}
\mathbb D_{r}(X|U) &\geq  \ushort {\mathbb D}_{r}(X|U) \label{eq:slb} %
\triangleq 
nN(X|U) \exp \left( - 2 r / n \right), 
\end{align}
with equality if $X = U + S$, where $S \sim \mathcal N(0, \sigma^2 \mat I)$. 
\end{thm}
\begin{proof} \thmref{thm:slb} is a conditional version of Shannon's lower bound \cite{shannon1959coding}. 
Using Propositions \ref{prop:gauss} and \ref{prop:unconditioning}, we can write, for any $Y$ such that $I(X; Y | U) \leq r$, 
\begin{align}
\E{\|X - Y\|^2} &\geq \E{\| X - \E{X|Y, U}\| ^2} \label{eq:estineq}\\
&\geq N(X|Y, U)\\
&\geq N(X | U) \exp( - 2 r / n ). \label{eq:uncondapp}
\end{align}
The equality condition is verified by checking that $Y$ such that $X = Y + Z$, where $Z \sim \mathcal N(0, \sigma^2 \exp(- 2 r / n ) \mat I )$, attains equalities in \eqref{eq:estineq}--\eqref{eq:uncondapp}.
\end{proof}
 Shannon's lower bound is equal to the distortion-rate function of a white Gaussian vector with the same differential entropy as the original vector. Although beyond Gaussian $X$, Shannon's lower bound is rarely attained with equality \cite{gerrish1964nongaussian}, it is approached at high rates \cite{linkov1965evaluation}. The tightness of Shannon's lower bound at high rates is key to arguing that the bound in Theorem \ref{thm:main} can in fact be approached.

For convenience, we record the following result, which is an immediate corollary to \propref{prop:unconditioning}. 
\begin{prop}
\label{prop:Drc2}
Let $X \in \mathbb R^n$ be a random vector. The following inequality holds:  
\begin{align}
 \min_{{U} \in \mathbb R^n \colon I(X; U)  \leq s }   \ushort{\mathbb D}_{ r }( X | U ) &\geq \ushort{\mathbb D}_{ r + s } (X)\label{eq:Drc2}.
\end{align}
\end{prop}

\begin{remark}
If $X$ is Gaussian, then 
\begin{equation}
 \min_{{U} \in \mathbb R^n \colon I(X; U)  \leq s }  {\mathbb D}_{ r }( X | U ) = {\mathbb D}_{ r + s } (X), \label{eq:gauss}
\end{equation}
and the minimum is attained by a Gaussian $U$. For non-Gaussian $X$, $\geq$ holds in \eqref{eq:gauss}. 
\end{remark}

We are now equipped to prove our converse theorems.
\begin{proof}[Proof of \thmref{thm:seqslb}]
For any causal kernel $P_{\hat S^t \| S^t}$ induced by a code, denote the per-stage information rates
\begin{equation}
r_i \triangleq I(S^i; \hat S_i | \hat S^{i-1}). \label{eq:ri}
\end{equation}
Using $r_i \geq  I(S_i; \hat S_i | \hat S^{i-1})$, Shannon's lower bound (Theorem \ref{thm:slb}), \propref{prop:DrM} and \eqref{eq:epscale}, we lower-bound  the distortion at step $i$ as   
\begin{align}
\!\!\!\!\!\! \E{(S_i - \hat S_i)^T \mat W_i (S_i - \hat S_i) } \geq&~  {\mathbb D}_{ r_i,\, \mathsf W_i }( S_{i}| \hat S^{i - 1}) \label{eq:drlb1}\\
\geq&~  w_i\, \ushort d_i, \label{eq:drlb2}
\end{align}
where we denoted for brevity
\begin{align}
 \ushort d_i &\triangleq \ushort {\mathbb D}_{ r_i }( S_{i}| \hat S^{i - 1}), ~ i = 1, 2, \ldots, t, \\
 w_i &\triangleq (\det \mat W_i)^{\frac 1 n},
\end{align}
Next,  define $\ushort d_{0}$ satisfy $ a^2 \ushort d_{0} + n N(V) = n N(S_1)$. We establish the following recursion ($1 \leq i \leq t$):
\begin{align}
 \ushort d_i &= \ushort {\mathbb D}_{ r_i}( \mathsf A S_{i-1} + \vect{V}_{i-1}  | \hat S^{i - 1}) \label{eq:step2a0}\\
 &\geq \ushort {\mathbb D}_{  r_i  } ( \mathsf A S_{i-1} | \hat S^{i-1}) + \ushort {\mathbb D}_{  r_i }(V)  \label{eq:step2a} \\
&\geq \ushort {\mathbb D}_{  r_{i-1} + r_i}(  \mathsf A S_{i-1}  | \hat S^{i-2}) + \ushort {\mathbb D}_{  r_i}(V) \label{eq:step2aa}\\
&= \left( a^2 \ushort d_{i-1} +n N(V) \right) \exp(-2r_i / n)\label{eq:recur}
\end{align}
where \eqref{eq:step2a} is by the conditional EPI (\thmref{thm:epi}), \eqref{eq:step2aa} is due to  \eqref{eq:Drc2}.  Note that \eqref{eq:step2a0}--\eqref{eq:recur} holds for an arbitrary encoded sequence $\hat S_1, \ldots, \hat S_{i-1}$, including the optimal one.

Rewriting \eqref{eq:recur} as
\begin{align}
2 r_i / n &\geq  \log ( a^2 \ushort d_{i-1} + n N(V)) - \log \ushort d_{i}, \label{eq:recurr}
\end{align}
we deduce 
\begin{align}
&~\!\!\!\!\!\!\!\!\!\! \frac 1 n \sum_{i = 1}^{t} r_i \geq \frac 1 n \sum_{i = t_\epsilon}^{t} r_i \geq \label{eq:sumlbeps}\\
&~\!\!\!\!\!\!\!\!\!\!  \frac {1} 2 \log \frac{a^2 \ushort d_{t_{\epsilon} - 1} + n N(V)} {a^2 \ushort d_{t} + n N(V)}   + \frac 1 2 \sum_{i = t_{\epsilon}}^{t} \log \left( a^2 + \frac {N(V)}{\ushort d_i / n}\right) \!\!\!\!\!\!\!\label{eq:sumr},
\end{align}
where $t_{\epsilon}$ is defined for any $\epsilon > 0$ as the smallest number such that for all $i \geq t_{\epsilon}, ~ w_i \geq  w - \epsilon$. 
Assumption \eqref{eq:liminfass} ensures that $t_{\epsilon} < \infty$. The distortion constraint in \eqref{eq:causalRd} and the bound \eqref{eq:drlb2} imply 
\begin{align}
 \frac {w - \epsilon} t \textstyle{\sum_{i = t_{\epsilon}}^{t} } \ushort d_i &\leq 
\frac 1 t \textstyle{\sum_{i = 1}^{t}}  w_i \, \ushort d_i \leq d \label{consrt}.
\end{align}
In particular, \eqref{consrt} implies that $\ushort d_{t} \leq d \frac{t}{w - \epsilon}$, which together with $a^2 \ushort d_{t_{\epsilon} - 1} + n N(V) \geq n \min\{N(S_1), N(V)\}$  means that the first term in \eqref{eq:sumr} normalized by $t$ is bounded below by a quantity that vanishes as $t \to \infty$. 

Since the function $x \mapsto \log \left( a^2 + \frac {n N(V)}{x}\right)$ is convex and decreasing, by Jensen's inequality and \eqref{consrt} the sum in \eqref{eq:sumr} is bounded below as 
\begin{align}
&~ \sum_{i = t_{\epsilon}}^{t} \log \left( a^2 + \frac {N(V)}{\ushort d_i / n}\right) \notag\\
\geq&~ (t - t_{\epsilon})   \log \left( a^2 + \frac  {(w - \epsilon) N(V)}{d / n }  \frac{t - t_{\epsilon}}{t} \right), \label{eq:jensent}
\end{align}
Diving both sides of \eqref{eq:jensent} by $t$ and taking taking a $\lim_{t \to \infty}$ followed by a $\lim_{\epsilon \to 0}$, we obtain \eqref{eq:slbseq}.
\end{proof}

\begin{proof}[Proof of Theorem \ref{thm:seqslb+}]
We start by making two observations. 
First, putting $S_i^\prime \triangleq \mat J^{-1 }S_{i}$ and $V_i^\prime \triangleq \mat J^{-1} V_i$,  we may write
\begin{equation}
 S_{i+1}^\prime =  {\mathsf A^\prime} S_i^\prime + V_{i}^\prime, \label{eq:Stprime}
\end{equation}
Causal rate-distortion functions of $S^\infty$ and $S^{\prime \infty}$ satisfy
\begin{equation}
 \mathbb R_{S^{\infty}, \{\mat W_i\}}(d) = \mathbb R_{S^{\prime\, \infty}, \{\mat J^{T}  \mat W_i \mat J\}}(d), \label{eq:rdj}
\end{equation}
where we indicated the weight matrices in the subscript.

Second,  if $0 \preceq \mat \Pi^T \mat \Pi \preceq \mat I$,  and $\mat \Pi$ commutes with $\mat L$, then
\begin{equation}
 \mat \Pi^T\, \mat L^T \mat L\, \mat \Pi \preceq \mat L^T \mat L. \label{eq:commute}
\end{equation}

Due to \eqref{eq:rdj} we may focus on evaluating the rate-distortion function for $S^{\prime\, \infty}$.
Since $\left( \mat \Pi_{\ell} \mat \Pi_{\ell}^T \right)^2 \preceq \mat \Pi_{\ell} \mat \Pi_{\ell}^T \preceq \mat I$ and $\mat \Pi_{\ell} \mat \Pi_{\ell}^T$ commutes with $\mat L_i$, we may apply   \eqref{eq:commute} and Theorem \ref{thm:slb} to obtain    
\begin{align}
&~  \E{(S_i^\prime - \hat S_i^\prime)^T  \mat J^{T}  \mat W_i \mat J  (S_i^\prime - \hat S_i^\prime) } \label{eq:drca}\\
\geq&~ \E{(S_i^\prime - \hat S_i^\prime)^T  \mat \Pi_{\ell} \mat \Pi_{\ell}^T \mat J^{T}  \mat W_i \mat J  \mat \Pi_{\ell} \mat \Pi_{\ell}^T (S_i^\prime - \hat S_i^\prime) }\\
\geq&~  w_i^\prime \, \ushort d_i^\prime, \label{eq:drcb}
\end{align}
where
\begin{align}
\ushort d_i^\prime &\triangleq \ushort {  \mathbb D}_{ r_i }( \mat \Pi_{\ell}^T S_{i}^\prime | \hat S^{i - 1}), ~ i = 1, 2, \ldots, t.
\end{align}

Since $\mat \Pi_{\ell} \mat \Pi_{\ell}^T$ commutes with ${\mathsf A}^\prime$ and $\mat \Pi_{\ell}^T \mat \Pi_{\ell} \mat \Pi_{\ell}^T = \mat \Pi_{\ell}^T$, \begin{equation}
 \mat \Pi_{\ell}^T \mat A^\prime = \mat \Pi_{\ell}^T \mat A^\prime \mat \Pi_{\ell} \mat \Pi_{\ell}^T \label{eq:proj}.
\end{equation}
 Using \eqref{eq:proj}, \thmref{thm:epi} and \propref{prop:Drc2}, we establish 
\begin{align}
  \ushort d_i^\prime &=  \ushort {\mathbb D}_{ r_i}( \mat \Pi_{\ell}^T \mat A^\prime S_{i-1}^\prime + \mat \Pi_{\ell}^T \vect{V}_{i-1}^\prime  | \hat S^{i - 1}) \label{eq:step2b0}\\
  &\geq \ushort {\mathbb D}_{ r_i}( \mat \Pi_{\ell}^T \mat A^\prime S_{i-1}^\prime  | \hat S^{i - 1})  + \ushort {\mathbb D}_{ r_i}(  \mat \Pi_{\ell}^T \vect{V}_{i-1}^\prime )  \\
  &= \ushort {\mathbb D}_{ r_i}( \mat \Pi_{\ell}^T \mat A^\prime \mat \Pi_{\ell} \mat \Pi_{\ell}^T S_{i-1}^\prime  | \hat S^{i - 1})  + \ushort {\mathbb D}_{ r_i}(  \mat \Pi_{\ell}^T \vect{V}_{i-1}^\prime )  \\
 &=  a^{\prime\,2} \ushort {\mathbb D}_{  r_i  } ( \mat \Pi_{\ell}^T S_{i-1}^\prime | \hat S^{i-1}) + \ushort {\mathbb D}_{  r_i }( \mat \Pi_{\ell}^T \vect{V}_{i-1}^\prime )  \\
&\geq \left( a^{\prime\,2} \ushort d_{i-1}^\prime + m N(\mat \Pi_{\ell}^T V^\prime) \right) \exp\left( -2r_i/n \right)   \label{eq:recurb}
\end{align}
The rest of the proof follows that of \thmref{thm:seqslb}.
\end{proof}

\begin{proof}[Proof of \thmref{thm:seqslb++}]
It is easy to see (along the lines of \eqref{eq:sumlbeps}) that if $\mat W_i \to \mat L^T \mat L$, we may put $\mat W_i \equiv \mat L^T \mat L$ without affecting the (causal) rate-distortion function. Similarly, if $V_i = \mat K_i V_i^\prime$ and $\mat K_i \to \mat K$, we may put $V_i = \mat K V_i^\prime$. We will therefore focus on bounding the distortion with weight matrix $\mat L^T \mat L$ and with $V_i = \mat K V_i^\prime$.

Adopting the convention $S_1 \equiv V_0$, we rewrite \eqref{eq:St} as, 
\begin{align}
S_i = %
  \sum_{j = 0}^{i-1} \mat A^{i- j-1} V_{j}
 \label{eq:Siassum}
\end{align}
Fixing causal reproduction vector $\hat S^i$, for $0 \leq j \leq i - 1$, consider the random variable
\begin{align}
\tilde V_j &\triangleq V_j - \E{V_j | \hat S^i} %
= V_j - \E{V_j | \hat S_{j+1}^i}, \label{eq:causalcut}
\end{align}
where \eqref{eq:causalcut} holds because $V_j$ is independent of $\hat S^{j}$. Note that different  $\tilde V_j$'s are uncorrelated. 
Indeed, to verify that $V_j$ and $V_{j+\ell}$ are uncorrelated, note that since $V_{j+\ell} \pperp \hat S_{j + 1}^{j+\ell}$ and $V_{j+\ell} \pperp V_j$, we have 
$\E{\tilde V_{j+\ell} | V_j, \hat S_{j+1}^i} = \mathbf 0$, and thus 
\begin{align}
 \E{ \tilde V_i \tilde V_{j+\ell}^T} =  \E{\tilde V_i\, \E{\tilde V_{j+\ell}^T | V_i, \hat S_{j+1}^i}} = \mat 0. \label{eq:Vjtildeuncorr}
\end{align}

Using \eqref{eq:Siassum} and \eqref{eq:Vjtildeuncorr}, we write
\begin{align}
  &~ \E{\left(S_i - \E{ S_i | \hat S^{i}} \right)^T  \mat L^T \mat L  \left(S_i - \E{ S_i | \hat S^{i}} \right) } \notag\\
  =&~ \textstyle{\sum_{j = 1}^{i}} \E{ \tilde V_{j-1}^T  \mat A^{i - j\, T}  \mat L^T \mat L \mat A^{i - j}  \tilde V_{j-1}  } \label{eq:uncorapp}\\
 \geq&~  \textstyle{\sum_{j = 1}^{i}} \mathbb D_{\sum_{\ell = 0}^{i-j} r_{i - \ell} } ( \mat L \mat A^{i - j} V_{j-1}) \label{eq:Drig}\\
 \geq&~  w N(V^\prime)^{\frac k m} \textstyle{\sum_{j = 1}^{i} a^{2(i - j)}} \exp\left( -\frac{2}{n} \textstyle{\sum_{\ell = j}^i} r_\ell  \right)   \label{eq:scelim} \\
 \triangleq&~ w \ushort d_i^{\prime\prime} ,
\end{align}
where \eqref{eq:uncorapp} uses \eqref{eq:Siassum} and \eqref{eq:Vjtildeuncorr};  \eqref{eq:Drig} leverages \propref{prop:DrM} to minimize each term of the sum over $P_{\hat S_{j}^i | V_{j-1}}$ subject to the constraint 
\vspace{-.9em}
\begin{align}
 I(V_{j-1}; \hat S_{j}^i) \leq \sum_{\ell = j}^i r_\ell, \label{eq:risum}
\end{align}
where $r_\ell$ is the per-stage rate, as defined before in \eqref{eq:ri}; \eqref{eq:scelim} is due to \propref{prop:hlx}. To verify that constraint \eqref{eq:ri} implies \eqref{eq:risum}, we apply the independence of $V_{j-1}$ and $ \hat S^{j-1}$ and the chain rule of mutual information to write
\begin{align}
 & I(V_{j-1}; \hat S_{j}^i) =  I(V_{j-1}; \hat S_{j}^i | \hat S^{j-1})
  \leq  I(S^j; \hat S_{j}^i | \hat S^{j-1})\\
 &= \sum_{\ell = j}^{i} I(S^{j}; \hat S_{\ell} | \hat S^{\ell - 1} ) 
  \leq \sum_{\ell = j}^{i} I(S^\ell; \hat S_{\ell} | \hat S^{\ell - 1} ) 
  = \sum_{\ell = j}^i r_\ell.\notag
\end{align}

Finally, observe using \eqref{eq:scelim} that $\ushort d_{i+1}^{\prime\prime}$ and $\ushort d_i^{\prime\prime}$ are tied in a recursive relationship akin to that in \eqref{eq:recur}:
\begin{align}
\ushort d_{i+1}^{\prime\prime} &= \exp\left(-2 r_{i + 1} / n\right) \left(a^2  \ushort d_i^{\prime\prime} +  N(V^\prime)^{\frac k m} \right).
\end{align}
The rest of the proof follows along the lines of \eqref{eq:recurr}--\eqref{eq:jensent}. 
\end{proof}

\begin{proof}[Proof of \thmref{thm:main}]
If $\mat A$ is rank-deficient, the right side of \eqref{eq:main} is $-\infty$, and there is nothing to prove. Assume $\det \mat A \neq 0$. Further, the case $\rank \mathsf B < n$ implies $\det \mat M = 0$ and is covered by \thmref{thm:mainu}. Assume $\rank \mat B = n$.
According to \corref{cor:link}, \thmref{thm:seqslb} and \eqref{eq:controlledratecost}, it suffices to lower bound $\mathbb R_{S^\infty}(d)$ with weight matrices $\mat W_i$ in \eqref{eq:Wimat}.  Since $\mathsf S > 0$ (e.g. \cite{bertsekas1995dynamic}), it follows that $ \mat A^T \mat M \mat A > 0$ and thus \eqref{eq:liminfass} is satisfied with $w = \det (\mat A^T \mat M \mat A)$.  Therefore, \thmref{thm:seqslb} applies, and the result of \thmref{thm:main} is immediate. 
\end{proof}

\begin{proof}
[Proof of \thmref{thm:mainu}]
According to \corref{cor:link} and \eqref{eq:controlledratecost}, it suffices to lower-bound $\mathbb R_{S^\infty}(d)$ with weight matrices $\mat W_i$ in \eqref{eq:Wimat}. 
Consider first the case $\rank \mat B = n$. Let $\mathsf \Lambda_i$ be a diagonal matrix such that $0 \preceq \mat \Lambda_i \preceq \mathsf J^{T} \mathsf M_i \mathsf J$.
The weight matrices $\mathsf J^{T} \mat A^T \mathsf M_i  \mat A \mathsf J = \mat A^{\prime\, T} \mat M_i^\prime \mat A^\prime \succeq \mat A^{\prime\, T} \mat \Lambda_i \mat A^\prime$, and thus the assumption of \thmref{thm:seqslb+} is satisfied with $\mat V_i = \mat \Lambda_i^{\frac 1 2} \mat A^\prime$, and \eqref{eq:mainpi} follows. %

If $\rank \mathsf B < n$, $\mathsf M$ is singular,  and the bound in \eqref{eq:mainpi} reduces to simply 
\begin{equation}
 \mathbb R(b) \geq \ell \log a^\prime. \label{eq:lbnob}
\end{equation}
To show \eqref{eq:lbnob}, fix some $\epsilon > 0$. Without loss of generality, we assume $\rank \mathsf B = m < n$, and we augment~$\mathsf B$ as follows:%
\begin{equation}
\mathsf B_\epsilon \triangleq 
 \begin{bmatrix}
 \mathsf B & \epsilon \tilde {\mathsf B} 
 \end{bmatrix},
 \label{eq:Beps}
\end{equation}
where $(n - m) \times n$ matrix $\tilde{\mathsf B}$ is chosen so that the columns of $\mathsf B_\epsilon$ span $\mathbb R^n$. 
We also augment the $m \times m$ matrix $\mathsf R$ in~\eqref{eq:lqr}:
\begin{equation}
\tilde{\mathsf R} \triangleq   
\begin{bmatrix}
 \mathsf R & \mathsf 0\\
  \mathsf 0 & \mathsf I_{n - m}
 \end{bmatrix}.
 \label{eq:Reps}
\end{equation}
Consider the augmented system parameterized by $\epsilon$:
\begin{align}
\vect{X}_{i+1} &=  \mathsf A \vect{X}_i + \mathsf B_\epsilon \tilde{U}_i + \vect{V}_i, \label{eq:system1aug}
\end{align}
where control inputs $\tilde{U}_i$ are $n$-dimensional.  
The augmented system in \eqref{eq:system1aug} achieves the same or smaller quadratic cost as the system in \eqref{eq:system1}, because we can always let 
$
 \tilde{U}_i = 
 \begin{bmatrix}
 U_i & 
 \mathbf 0
  \end{bmatrix}^T
$
to equalize the costs. Therefore, 
\begin{align}
\mathbb R(b) &\geq \sup_{\epsilon > 0}{\mathbb R}_{\epsilon}(b)  \label{eq:rb2}, 
\end{align}
where ${\mathbb R}_{\epsilon}(b)$ denotes the rate-cost function for the system in \eqref{eq:system1aug} with parameter $\epsilon$ in \eqref{eq:Beps}, \eqref{eq:Reps}. In particular, since \eqref{eq:lbnob} holds for the augmented system it must also hold for the original system. 
\end{proof}

\begin{proof}[Proof of \thmref{thm:mainm}]
 According to \propref{prop:DrM}, the causal rate-distortion function of the uncontrolled process $\{S_i\}$ with weight matrices $\mat W_i = \mat A^T \mat M_i \mat A$ is equal to that of $\{\mat A S_i\}$ with weight matrices $\{\mat M_i\}$. Putting $S_i^{\prime\prime} \triangleq \mat A S_i$ and noticing that $S_{i+1}^{\prime\prime} = \mat A S_i^{\prime\prime} + \mat A V_i$, we apply  \thmref{thm:seqslb++} to conclude
\begin{equation}
 \mathbb R_{S^{\prime \prime \infty}}(d)  \geq \frac m 2 \log \left( a^{2} + \frac {\mu N(\mat A V)^{\frac n m}}{d/m}\right),
\end{equation}
and \thmref{thm:mainm} follows via \corref{cor:link} and \eqref{eq:controlledratecost}.
\end{proof}

\begin{proof}[Proof of \thmref{thm:mainpo}]
We assume $\rank \mat B = \rank \mat C = n$. The more general case is considered in \thmref{thm:mainpou}.

According to \eqref{eq:kalman} and \eqref{eq:covinno}, $\{\hat X_i\}$ is Gauss-Markov process, whose additive noise $\mathsf K_{i} \tilde Y_i$ has covariance  matrix
$\mat N_i \triangleq \mathsf K_{i} \left( \mathsf C  \mathsf P_{i|i-1}  \mathsf C^T + \mathsf \Sigma_W \right) \mat K_i^T$. By \thmref{thm:seqslb} and \eqref{eq:controlledratecost},
\begin{equation}
 \mathbb R_{\hat X^\infty\| \mathcal D U^{\infty}}(d) \geq \log |\det \mat A| + \frac n 2 \log \left(1  + \frac {\left(\det \mat N \mat M\right)^{\frac 1 n} }{d/n}\right), \notag
\end{equation}
and \thmref{thm:mainpo} follows immediately via \corref{cor:link}.  
\end{proof}

\begin{proof}[Proof of \thmref{thm:mainpou}]
The proof is similar to that of \thmref{thm:mainu} and uses \thmref{thm:seqslb+} to lower-bound $\mathbb R_{\hat X^\infty\| \mathcal D U^{\infty}}(d)$.  
\end{proof}

\begin{proof}[Proof of \thmref{thm:mainpom}]
The proof is similar to that of \thmref{thm:mainm}, and uses \thmref{thm:seqslb++} to lower-bound $\mathbb R_{\hat X^\infty\| \mathcal D U^{\infty}}(d)$. 
\end{proof}

\section{Achievability theorems: tools and proofs}
\label{sec:a}

In this section, we will prove Theorems \ref{thm:maina} and \ref{thm:mainpoa}.  
 
 Our achievability scheme employs lattice quantization. 
A  lattice $\mathcal C$ in $\mathbb R^n$ is a discrete set of points that is closed under reflection and addition.
The nearest-neighbor quantizer is the mapping $\mathsf q_{\mathcal C} \colon \mathbb R^n \mapsto \mathcal C$ defined by
\begin{equation}
\mathsf q_{\mathcal C} (x) \triangleq \argmin_{c \in \mathcal C} \|x - c\| \label{eq:q}.
\end{equation}
{\it Covering efficiency} of lattice $\cC$ is measured by 
\begin{align}
\rho_{\cC} &\triangleq \left( \frac { {B}_{\cC} }{V_{\cC} } \right)^{\frac 1 n} \label{eq:rho}, 
\end{align}
where $V_{\mathcal C}$ the volume of the Voronoi cells of lattice $\mathcal C$:
\begin{align}
 V_{\mathcal C} &\triangleq \mathrm{Vol} \left( \left\{ x \in \mathbb R^n \colon \mathsf q_{\mathcal C} (x) = c \right\} \right),
 \end{align}
where arbitrary $c \in \mathcal C$, and ${B}_{\cC}$ is the volume of a ball whose radius is equal to that of the Voronoi cells of $\mathcal C$. The radius of  ${B}_{\cC}$ is called {\it covering radius} of lattice $\cC$.
 By definition, $
 \rho_{\cC} \geq 1$,
and the closer $\rho_{\cC}$ is to $1$ the more sphere-like the Voronoi cells of $\cC$ are and the better lattice $\cC$ is for covering.

\begin{proof}[Proof of \thmref{thm:slba}]
The proof analyses a DPCM scheme. First, we describe how the codebook is generated, then we describe the operation of the encoder and the decoder, and then we proceed to the analysis of the scheme. 

{\it Codebook design.}  To maximize covering efficiency, we use the best known $n$-dimensional lattice quantizer $\mathsf q = \mathsf q_{\mathcal C^n}$ scaled so that its covering radius is $\leq \sqrt d$. 

{\it Encoder.} 
Upon observing $S_i$, the encoder computes the state innovation $\tilde S_i$ recursively using the formula 
\begin{equation}
\tilde S_i \triangleq S_i - \mat A \hat S_{i - 1}, \label{eq:Stilde}
\end{equation}
where $\hat S_i$ is the decoder's state estimate at time $i$ (put $\hat S_{0} \triangleq 0$). The encoder transmits the index of 
\begin{equation}
Q_i \triangleq  \mathsf q(\mat W_i^{\frac 1 2} \tilde S_i). 
\end{equation}

{\it Decoder.} 
The decoder recovers the lattice cell identified by the encoder, and forms its state estimate as
\begin{align}
 \hat S_{i} &= \mat A \hat S_{i-1} + \hat{\tilde S}_{i}, \\
 \hat{\tilde S}_i &\triangleq \mat W_i^{-\frac 1 2} Q_i.
\end{align}

{\it Analysis.} 
The distortion at step $i$ is given by 
\begin{align}
\left(  S_i - \hat S_i \right)^T
 \mat W_i 
\left(S_i -  \hat S_i \right) \notag
=&~ 
\left(  \tilde S_i - \hat{\tilde S}_i  \right)^T
\mat W_i
 \left(  \tilde S_i - \hat{\tilde S}_i  \right) 
\\
 =& ~
\left\|  \mat W_i^{\frac 1 2} \tilde S_i -  Q_i \right\|^2
 \leq  d.  \label{eq:amseu}
\end{align}
It remains to upper-bound the entropy of $Q^t$. Since $H(Q^t) \leq \sum_{i = 1}^t H(Q_i)$, it suffices to bound the unconditional entropy of $Q_i$.
First, we establish that $\mat W_i^{\frac 1 2} \tilde S_i$ has a regular density.
Using the assumption that $V_i$ has a $(c_0, c_1)$-regular density, it's easy to see that $\mat W_i^{\frac 1 2} V_i$ has $\left( w_i^{-1}c_0, w_i^{-1} c_1\right)$-regular density, where $w_i$ is the minimum eigenvalue of $\mat W_i$. 
Furthermore,  similar to \eqref{eq:amseu},
\begin{equation}
 \left(  S_i - \hat S_i \right)^T
\mat A^T \mat W_i \mat A
\left( S_i -  \hat S_i \right) 
\leq a_i d, \label{eq:amseua}
\end{equation}  
 where
 $a_i$ is the following operator norm of $\mathsf A$: 
\begin{equation}
a_i \triangleq \sup_{z \neq \mathbf 0}  \frac{ z ^T {\mathsf A}^T
 \mat W_i
\mathsf A  z  }{ z^T \mat W_i z}. \label{eq:ainorm}
\end{equation}
From \eqref{eq:Stilde} and \eqref{eq:St},
\begin{equation}
 \tilde S_i = \mathsf A(S_{i-1} - \hat S_{i-1}) + V_{i-1}, \label{eq:sinno}
\end{equation}
and it follows via \cite[Prop. 3]{polyanskiy2015wasserstein} that $\mat W_i^{\frac 1 2} \tilde S_i$ has 
$(w_i^{-1}(c_0 + a_i^{\frac 1 2} d^{\frac 1 2}  c_1), w_i^{-1} c_1)$-regular density. 

Combining \eqref{eq:amseua} and \eqref{eq:sinno} yields
\begin{align}
 \Var{ \mat W_i^{\frac 1 2} \tilde S_i} &\leq a_i d + v_i, 
\end{align}
where we denoted for brevity
\begin{align}
v_i \triangleq  \tr\left(\mathsf \Sigma_V \mat W_i \right), 
\end{align}

Now, \cite[Th. 8]{kostina2016lowd} implies that the entropy of $Q_i$ satisfies:
\begin{align}
 &~ H\left( Q_i \right) \leq \min_{ \tilde d \leq d }\Bigg\{ \frac n 2 \log \frac{N(\mat W_i^{\frac 1 2} \tilde S_i )}{{\tilde d}/n} 
 \notag\\
 &
 +  2 \frac{{\tilde d}^{\frac 1 2}}{w_i} \log e 
 \cdot \left( c_1 \left( a_i d + v_i \right)^{\frac 1 2} + c_0 + c_1 \left( 1 + a_i^{\frac 1 2} \right){\tilde d}^{\frac 1 2} \right) \Bigg\}  \notag\\
 &+ \alpha_n + n \log \rho_{\mathcal C_n}, \label{eq:Hqws}
 \end{align}
 where $\alpha_n$ and $n \log \rho_{\mathcal C_n}$ are of order $\bigo{\log n}$.  

To estimate the entropy power of $\mat W_i^{\frac 1 2} \tilde S_i$, we use \eqref{eq:amseua} and \eqref{eq:sinno} to bound the Wasserstein distance between $\mat W_i^{\frac 1 2} \tilde S_i$ and $\mat W_i^{\frac 1 2} V_{i-1}$, so that \cite[Prop. 1]{polyanskiy2015wasserstein} applies to yield:
\begin{align}
 h( \mat W_i^{\frac 1 2} \tilde S_i ) &\leq h(\mat W_i^{\frac 1 2} V_{i-1}) + \log e \frac{ \left( a_i d_i\right)^{\frac 1 2}}{w_i}  \label{eq:hqws}
 \\
 &\phantom{=}
 \cdot  \left(\frac {c_1}{2} v_i^{\frac 1 2} + \frac {c_1} 2 \left( a_i d + v_i\right)^{\frac 1 2} + c_0\right).   \notag
\end{align}
Combining \eqref{eq:Hqws} and \eqref{eq:hqws}, we conclude that
\begin{align}
\label{eq:boundHQi}
 &~ H\left( Q_i \right) \leq \\
 &~ \min_{ \tilde d \leq d } \Bigg\{ \frac n 2 \log \frac{N( \mat W_i^{\frac 1 2} V_{i-1})}{\tilde d/n} 
 +  \beta_i(\tilde d ) \Bigg\} + \alpha_n + n \log \rho_{\mathcal C_n} , \notag
 \end{align}
where  
$\beta_i(d) = \bigo{d^{\frac 1 2} }$ is given by
\begin{align}
\beta_i(d) \triangleq&~  
\frac{d^{\frac 1 2}}{w_i} \log e \Bigg( \frac 1 2 c_1a_i^{\frac 1 2}  v_i^{\frac 1 2} + c_0\left( 2 + a_i^{\frac 1 2}\right) \\
&+ c_1 \left( 2 + \frac{a_i^{\frac 1 2}  }{2} \right) \left( a_i d + v_i \right)^{\frac 1 2} + 2 c_1 \left( 1 + a_i^{\frac 1 2} \right) d^{\frac 1 2}\Bigg). \notag
\end{align}
Recalling \eqref{eq:epscale} and using the resulting bound \eqref{eq:boundHQi} to bound $\lim_{t \to \infty} \frac 1 t \sum_{i = 1}^t H(Q_i)$, we obtain the statement of \thmref{thm:slba}. 
\end{proof}

\begin{proof}[Proof of Theorem \ref{thm:maina}] Due to \corref{cor:link} and \eqref{eq:controlledhcost}, it suffices to bound the entropy-distortion function of the process \eqref{eq:St}. Such a bound is provided in \thmref{thm:slba}.
\end{proof}

\begin{proof}[Proof of Theorem \ref{thm:mainpoa}]
 Due to \corref{cor:link}, it suffices to bound the conditional entropy-distortion function of the Kalman filter estimates process in \eqref{eq:kalman}. Such a bound follows from \eqref{eq:controlledhcost} and \thmref{thm:slba}.
\end{proof}

\section{Conclusion}
We studied the fundamental tradeoff between the communication requirements and the attainable quadratic cost in fully and partially observed linear stochastic control systems. We introduced the rate-cost function in Definition~\ref{defn:rb}, and showed sharp lower bounds to it in Theorems \ref{thm:main}, \ref{thm:mainu}, \ref{thm:mainm} (fully observed system) and Theorems \ref{thm:mainpo}, \ref{thm:mainpou}, \ref{thm:mainpom} (partially observed system). The achievability results in \thmref{thm:maina} (fully observed system) and \thmref{thm:mainpoa} (partially observed system) show that the converse can be approached, in the high rate / low cost regime, by a simple variable-rate lattice-based scheme in which only the quantized value of the innovation is transmitted.  Via the separation principle, the same conclusions hold for causal compression of Markov sources: a converse, which may be viewed as a causal counterpart of Shannon's lower bound, is stated in \thmref{thm:seqslb}, and a matching achievability in \thmref{thm:slba}.

Extending the analysis of the partially observed case to non-Gaussian noises would be of interest. 
It also remains an open question whether the converse bound in Theorem \ref{thm:main} can be approached by fixed-rate quantization, or over noisy channels.  Finally, it would be interesting to see whether using non-lattice quantizers can help to narrow down the gap in \vlong{Figures \ref{fig:rdg}, \ref{fig:rdl} and \ref{fig:rdpo}}{\figref{fig:rdl}}.

\section{Acknowledgement}
The authors acknowledge many stimulating discussions with Dr. Anatoly Khina and his helpful comments on the earlier versions of the manuscript. The authors are also grateful to Ayush Pandey, who  generated the \vlong{plots in Figures \ref{fig:rdg}, \ref{fig:rdl} and \ref{fig:rdpo}}{plot in \figref{fig:rdl}}.

\bibliographystyle{IEEEtran}
\bibliography{../../Bibliography/rateDistortion,../../Bibliography/control,../../Bibliography/vk}

\begin{thebibliography}{1}
\providecommand{\url}[1]{#1}
\csname url@samestyle\endcsname
\providecommand{\newblock}{\relax}
\providecommand{\bibinfo}[2]{#2}
\providecommand{\BIBentrySTDinterwordspacing}{\spaceskip=0pt\relax}
\providecommand{\BIBentryALTinterwordstretchfactor}{4}
\providecommand{\BIBentryALTinterwordspacing}{\spaceskip=\fontdimen2\font plus
\BIBentryALTinterwordstretchfactor\fontdimen3\font minus
  \fontdimen4\font\relax}
\providecommand{\BIBforeignlanguage}[2]{{%
\expandafter\ifx\csname l@#1\endcsname\relax
\typeout{** WARNING: IEEEtran.bst: No hyphenation pattern has been}%
\typeout{** loaded for the language `#1'. Using the pattern for}%
\typeout{** the default language instead.}%
\else
\language=\csname l@#1\endcsname
\fi
#2}}
\providecommand{\BIBdecl}{\relax}
\BIBdecl

\bibitem{rogers1964packing}
C.~A. Rogers, \emph{Packing and covering}.\hskip 1em plus 0.5em minus
  0.4em\relax Cambridge University Press, 1964, no.~54.

\bibitem{conway2013sphere}
J.~H. Conway and N.~J.~A. Sloane, \emph{Sphere packings, lattices and
  groups}.\hskip 1em plus 0.5em minus 0.4em\relax Springer Science \& Business
  Media, New York, 2013, vol. 290.

\end{thebibliography}


\begin{thebibliography}{10}
\providecommand{\url}[1]{#1}
\csname url@samestyle\endcsname
\providecommand{\newblock}{\relax}
\providecommand{\bibinfo}[2]{#2}
\providecommand{\BIBentrySTDinterwordspacing}{\spaceskip=0pt\relax}
\providecommand{\BIBentryALTinterwordstretchfactor}{4}
\providecommand{\BIBentryALTinterwordspacing}{\spaceskip=\fontdimen2\font plus
\BIBentryALTinterwordstretchfactor\fontdimen3\font minus
  \fontdimen4\font\relax}
\providecommand{\BIBforeignlanguage}[2]{{%
\expandafter\ifx\csname l@#1\endcsname\relax
\typeout{** WARNING: IEEEtran.bst: No hyphenation pattern has been}%
\typeout{** loaded for the language `#1'. Using the pattern for}%
\typeout{** the default language instead.}%
\else
\language=\csname l@#1\endcsname
\fi
#2}}
\providecommand{\BIBdecl}{\relax}
\BIBdecl

\bibitem{kostina2016control}
V.~Kostina and B.~Hassibi, ``Rate-cost tradeoffs in control,'' in
  \emph{Proceedings 54th Annual Allerton Conference on Communication, Control
  and Computing}, Monticello, IL, Oct. 2016, pp. 1157--1164.

\bibitem{tatikonda2004stochastic}
S.~Tatikonda, A.~Sahai, and S.~Mitter, ``Stochastic linear control over a
  communication channel,'' \emph{IEEE Transactions on Automatic Control},
  vol.~49, no.~9, pp. 1549--1561, 2004.

\bibitem{massey1990causality}
J.~Massey, ``Causality, feedback and directed information,'' in \emph{Proc.
  Int. Symp. Inf. Theory Applic.(ISITA-90)}, 1990, pp. 303--305.

\bibitem{kramer1998PhD}
G.~Kramer, ``Directed information for channels with feedback,'' Ph.D.
  dissertation, ETH Zurich, 1998.

\bibitem{gastpar2003tocodeornot}
M.~Gastpar, B.~Rimoldi, and M.~Vetterli, ``To code, or not to code: lossy
  source-channel communication revisited,'' \emph{IEEE Transactions on
  Information Theory}, vol.~49, no.~5, pp. 1147--1158, May 2003.

\bibitem{baillieul1999feedback}
J.~Baillieul, ``Feedback designs for controlling device arrays with
  communication channel bandwidth constraints,'' in \emph{ARO Workshop on Smart
  Structures, Pennsylvania State Univ}, 1999, pp. 16--18.

\bibitem{wong1999systems}
W.~S. Wong and R.~W. Brockett, ``Systems with finite communication bandwidth
  constraints. {II}. {S}tabilization with limited information feedback,''
  \emph{IEEE Transactions on Automatic Control}, vol.~44, no.~5, pp.
  1049--1053, 1999.

\bibitem{tatikonda2004control}
S.~Tatikonda and S.~Mitter, ``Control under communication constraints,''
  \emph{IEEE Transactions on Automatic Control}, vol.~49, no.~7, pp.
  1056--1068, 2004.

\bibitem{nair2007feedback}
B.~G.~N. Nair, F.~Fagnani, S.~Zampieri, and R.~J. Evans, ``Feedback control
  under data rate constraints: An overview,'' \emph{Proceedings of the IEEE},
  vol.~95, no.~1, pp. 108--137, 2007.

\bibitem{nair2004stabilizability}
G.~N. Nair and R.~J. Evans, ``Stabilizability of stochastic linear systems with
  finite feedback data rates,'' \emph{SIAM Journal on Control and
  Optimization}, vol.~43, no.~2, pp. 413--436, 2004.

\bibitem{brockett2000quantized}
R.~W. Brockett and D.~Liberzon, ``Quantized feedback stabilization of linear
  systems,'' \emph{IEEE transactions on Automatic Control}, vol.~45, no.~7, pp.
  1279--1289, 2000.

\bibitem{yuksel2010fixedrate}
S.~Y\"uksel, ``Stochastic stabilization of noisy linear systems with fixed-rate
  limited feedback,'' \emph{IEEE Transactions on Automatic Control}, vol.~55,
  no.~12, pp. 2847--2853, 2010.

\bibitem{yuksel2014lqg}
------, ``Jointly optimal {LQG} quantization and control policies for
  multi-dimensional systems,'' \emph{IEEE Transactions on Automatic Control},
  vol.~59, no.~6, pp. 1612--1617, 2014.

\bibitem{witsenhausen1979structure}
H.~S. Witsenhausen, ``On the structure of real-time source coders,'' \emph{The
  Bell System Technical Journal}, vol.~58, no.~6, pp. 1437--1451, 1979.

\bibitem{gaarder1982optimal}
N.~Gaarder and D.~Slepian, ``On optimal finite-state digital transmission
  systems,'' \emph{IEEE Transactions on Information Theory}, vol.~28, no.~2,
  pp. 167--186, 1982.

\bibitem{walrand1983optimal}
J.~C. Walrand and P.~Varaiya, ``Optimal causal coding-decoding problems,''
  \emph{IEEE Transactions on Information Theory}, vol.~29, no.~6, pp. 814--820,
  1983.

\bibitem{borkar2001optimal}
V.~S. Borkar, S.~K. Mitter, and S.~Tatikonda, ``Optimal sequential vector
  quantization of {M}arkov sources,'' \emph{SIAM journal on control and
  optimization}, vol.~40, no.~1, pp. 135--148, 2001.

\bibitem{teneketzis2006structure}
D.~Teneketzis, ``On the structure of optimal real-time encoders and decoders in
  noisy communication,'' \emph{IEEE Transactions on Information Theory},
  vol.~52, no.~9, pp. 4017--4035, 2006.

\bibitem{linder2014zerodelay}
T.~Linder and S.~Y{\"u}ksel, ``On optimal zero-delay coding of vector {M}arkov
  sources,'' \emph{IEEE Transactions on Information Theory}, vol.~60, no.~10,
  pp. 5975--5991, 2014.

\bibitem{wood2015walrand}
R.~G. Wood, T.~Linder, and S.~Y{\"u}ksel, ``Optimality of {W}alrand-{V}araiya
  type policies and approximation results for zero delay coding of {M}arkov
  sources,'' in \emph{Proceedings 2015 IEEE International Symposium on
  Information Theory}, Hong Kong, June 2015.

\bibitem{yuksel2008ita}
S.~Y\"uksel, T.~Basar, and S.~P. Meyn, ``Optimal causal quantization of
  {M}arkov sources with distortion constraints,'' in \emph{Information Theory
  and Applications Workshop, 2008}, 2008, pp. 26--30.

\bibitem{elia2001stabilization}
N.~Elia and S.~K. Mitter, ``Stabilization of linear systems with limited
  information,'' \emph{IEEE transactions on Automatic Control}, vol.~46, no.~9,
  pp. 1384--1400, 2001.

\bibitem{gorbunov1974prognostic}
A.~Gorbunov and M.~S. Pinsker, ``Prognostic epsilon entropy of a {G}aussian
  message and a {G}aussian source,'' \emph{Problemy Peredachi Informatsii},
  vol.~10, no.~2, pp. 5--25, 1974.

\bibitem{charalambous2014nonanticipative}
C.~D. Charalambous, P.~A. Stavrou, and N.~U. Ahmed, ``Nonanticipative rate
  distortion function and relations to filtering theory,'' \emph{IEEE
  Transactions on Automatic Control}, vol.~59, no.~4, pp. 937--952, 2014.

\bibitem{charalambous2008LQGoptimality}
C.~D. Charalambous and A.~Farhadi, ``{LQG} optimality and separation principle
  for general discrete time partially observed stochastic systems over finite
  capacity communication channels,'' \emph{Automatica}, vol.~44, no.~12, pp.
  3181--3188, 2008.

\bibitem{charalambous2011memoryfeedback}
C.~D. Charalambous, C.~K. Kourtellaris, and C.~Hadjicostis, ``Optimal encoder
  and control strategies in stochastic control subject to rate constraints for
  channels with memory and feedback,'' in \emph{2011 50th IEEE Conference on
  Decision and Control and European Control Conference}, Dec 2011, pp.
  4522--4527.

\bibitem{shafieepoorfard2016rationally}
E.~Shafieepoorfard, M.~Raginsky, and S.~P. Meyn, ``Rationally inattentive
  control of {M}arkov processes,'' \emph{SIAM Journal on Control and
  Optimization}, vol.~54, no.~2, pp. 987--1016, 2016.

\bibitem{silva2011framework}
E.~I. Silva, M.~S. Derpich, and J.~Ostergaard, ``A framework for control system
  design subject to average data-rate constraints,'' \emph{IEEE Transactions on
  Automatic Control}, vol.~56, no.~8, pp. 1886--1899, 2011.

\bibitem{silva2016characterization}
E.~Silva, M.~Derpich, J.~Ostergaard, and M.~Encina, ``A characterization of the
  minimal average data rate that guarantees a given closed-loop performance
  level,'' \emph{IEEE Transactions on Automatic Control}, 2016.

\bibitem{tanaka2016isit}
T.~Tanaka, K.~H. Johansson, T.~Oechtering, H.~Sandberg, and M.~Skoglund, ``Rate
  of prefix-free codes in {LQG} control systems,'' in \emph{Proceedings 2016
  IEEE International Symposium on Information Theory}, Barcelona, Spain, July
  2016, pp. 2399--2403.

\bibitem{stavrou2017ITWKalman}
P.~A. Stavrou, J.~{\O}stergaard, C.~D. Charalambous, and M.~Derpich, ``An upper
  bound to zero-delay rate distortion via kalman filtering for vector gaussian
  sources,'' in \emph{Proceedings 2017 IEEE Information Theory Workshop (ITW)},
  Nov 2017, pp. 534--538.

\bibitem{derpich2012uppercausal}
M.~S. Derpich and J.~Ostergaard, ``Improved upper bounds to the causal
  quadratic rate-distortion function for {G}aussian stationary sources,''
  \emph{IEEE Transactions on Information Theory}, vol.~58, no.~5, pp.
  3131--3152, May 2012.

\bibitem{tanaka2017semidefinite}
T.~Tanaka, K.-K.~K. Kim, P.~A. Parrilo, and S.~K. Mitter, ``Semidefinite
  programming approach to {G}aussian sequential rate-distortion trade-offs,''
  \emph{IEEE Transactions on Automatic Control}, vol.~62, no.~4, pp.
  1896--1910, 2017.

\bibitem{tanaka2017lqg}
T.~Tanaka, P.~M. Esfahani, and S.~K. Mitter, ``{LQG} control with minimum
  directed information: Semidefinite programming approach,'' \emph{IEEE
  Transactions on Automatic Control}, 2017.

\bibitem{shannon1959coding}
C.~E. Shannon, ``{Coding theorems for a discrete source with a fidelity
  criterion},'' \emph{IRE Int. Conv. Rec.}, vol.~7, no.~1, pp. 142--163, Mar.
  1959, reprinted with changes in {\it Information and Decision Processes}, R.
  E. Machol, Ed. New York: McGraw-Hill, 1960, pp. 93-126.

\bibitem{shannon1948mathematical}
------, ``{A mathematical theory of communication},'' \emph{Bell Syst. Tech.
  J.}, vol.~27, pp. 379--423, 623--656, July and October 1948.

\bibitem{stam1959some}
A.~J. Stam, ``Some inequalities satisfied by the quantities of information of
  {F}isher and {S}hannon,'' \emph{Information and Control}, vol.~2, no.~2, pp.
  101--112, 1959.

\bibitem{gish1968asymptotically}
H.~Gish and J.~Pierce, ``Asymptotically efficient quantizing,'' \emph{IEEE
  Transactions on Information Theory}, vol.~14, no.~5, pp. 676--683, 1968.

\bibitem{ziv1985universal}
J.~Ziv, ``On universal quantization,'' \emph{IEEE Transactions on Information
  Theory}, vol.~31, no.~3, pp. 344--347, 1985.

\bibitem{zamir1996onlatticenoise}
R.~Zamir and M.~Feder, ``On lattice quantization noise,'' \emph{IEEE
  Transactions on Information Theory}, vol.~42, no.~4, pp. 1152--1159, Jul.
  1996.

\bibitem{gersho1979asymptotically}
A.~Gersho, ``Asymptotically optimal block quantization,'' \emph{IEEE
  Transactions on Information Theory}, vol.~25, no.~4, pp. 373--380, Jul. 1979.

\bibitem{zamir1992universal}
R.~Zamir and M.~Feder, ``On universal quantization by randomized
  uniform/lattice quantizers,'' \emph{IEEE Transactions on Information Theory},
  vol.~38, no.~2, pp. 428--436, Mar. 1992.

\bibitem{linder1994tessellating}
T.~T. Linder and K.~K. Zeger, ``Asymptotic entropy-constrained performance of
  tessellating and universal randomized lattice quantization,'' \emph{IEEE
  Transactions on Information Theory}, vol.~40, no.~2, pp. 575--579, Mar. 1994.

\bibitem{kostina2016lowd}
V.~Kostina, ``Data compression with low distortion and finite blocklength,''
  \emph{IEEE Transactions on Information Theory}, vol.~63, no.~7, pp.
  4268--4285, July 2017.

\bibitem{fu2012lack}
M.~Fu, ``Lack of separation principle for quantized linear quadratic {G}aussian
  control,'' \emph{IEEE Transactions on Automatic Control}, vol.~57, no.~9, pp.
  2385--2390, 2012.

\bibitem{polyanskiy2015wasserstein}
Y.~Polyanskiy and Y.~Wu, ``Wasserstein continuity of entropy and outer bounds
  for interference channels,'' \emph{IEEE Transactions on Information Theory},
  vol.~62, no.~7, pp. 3992--4002, July 2016.

\bibitem{kim2008coding}
Y.-H. Kim, ``A coding theorem for a class of stationary channels with
  feedback,'' \emph{IEEE Transactions on Information Theory}, vol.~54, no.~4,
  pp. 1488--1499, 2008.

\bibitem{tatikonda2009feedbackcapacity}
S.~Tatikonda and S.~Mitter, ``The capacity of channels with feedback,''
  \emph{IEEE Transactions on Information Theory}, vol.~55, no.~1, pp. 323--349,
  Jan 2009.

\bibitem{bansal1989simultaneous}
R.~Bansal and T.~Ba{\c{s}}ar, ``Simultaneous design of measurement and control
  strategies for stochastic systems with feedback,'' \emph{Automatica},
  vol.~25, no.~5, pp. 679--694, 1989.

\bibitem{freudenberg2010stabilization}
J.~S. Freudenberg, R.~H. Middleton, and V.~Solo, ``Stabilization and
  disturbance attenuation over a {G}aussian communication channel,'' \emph{IEEE
  Transactions on Automatic Control}, vol.~55, no.~3, pp. 795--799, 2010.

\bibitem{khina2016multi}
A.~Khina, G.~M. Pettersson, V.~Kostina, and B.~Hassibi, ``Multi-rate control
  over {AWGN} channels: An analog joint source--channel coding perspective,''
  in \emph{55th IEEE Conference on Decision and Control}.\hskip 1em plus 0.5em
  minus 0.4em\relax Las Vegas, NV, Dec. 2016.

\bibitem{wyner1972upper}
A.~Wyner, ``An upper bound on the entropy series,'' \emph{Information and
  Control}, vol.~20, no.~2, pp. 176--181, Mar. 1972.

\bibitem{alon1994lower}
N.~Alon and A.~Orlitsky, ``A lower bound on the expected length of one-to-one
  codes,'' \emph{IEEE Transactions on Information Theory}, vol.~40, no.~5, pp.
  1670--1672, Sep. 1994.

\bibitem{szpankowski2011minimum}
W.~Szpankowski and S.~Verd{\'u}, ``{Minimum expected length of
  fixed-to-variable lossless compression without prefix constraints: memoryless
  sources},'' \emph{IEEE Transactions on Information Theory}, vol.~57, no.~7,
  pp. 4017--4025, July 2011.

\bibitem{tanaka2015stationary}
T.~Tanaka, ``Semidefinite representation of sequential rate-distortion function
  for stationary {G}auss-{M}arkov processes,'' in \emph{Proceedings 2015 IEEE
  Conference on Control Applications (CCA)}, Sep. 2015, pp. 1217--1222.

\bibitem{fischer1982optimalquantizedcontrol}
T.~Fischer, ``Optimal quantized control,'' \emph{IEEE Transactions on Automatic
  Control}, vol.~27, no.~4, pp. 996--998, Aug 1982.

\bibitem{bertsekas1995dynamic}
D.~P. Bertsekas, \emph{Dynamic programming and optimal control}.\hskip 1em plus
  0.5em minus 0.4em\relax Athena Scientific Belmont, MA, 1995, vol.~1.

\bibitem{tatikonda2000thesis}
S.~Tatikonda, ``Control under communication constraints,'' Ph.D. dissertation,
  M.I.T., 2000.

\bibitem{cover2012elements}
T.~M. Cover and J.~A. Thomas, \emph{Elements of information theory},
  2nd~ed.\hskip 1em plus 0.5em minus 0.4em\relax John Wiley \& Sons, 2012.

\bibitem{gerrish1964nongaussian}
A.~Gerrish and P.~Schultheiss, ``Information rates of non-{G}aussian
  processes,'' \emph{IEEE Transactions on Information Theory}, vol.~10, no.~4,
  pp. 265--271, Oct. 1964.

\bibitem{linkov1965evaluation}
Y.~N. Linkov, ``Evaluation of $\epsilon$-entropy of random variables for small
  $\epsilon$,'' \emph{Problems of Information Transmission}, vol.~1, no.~2, pp.
  18--26, 1965.

\end{thebibliography}

\begin{IEEEbiography}
    [{\includegraphics[width=1in,height=1.25in,clip,keepaspectratio]{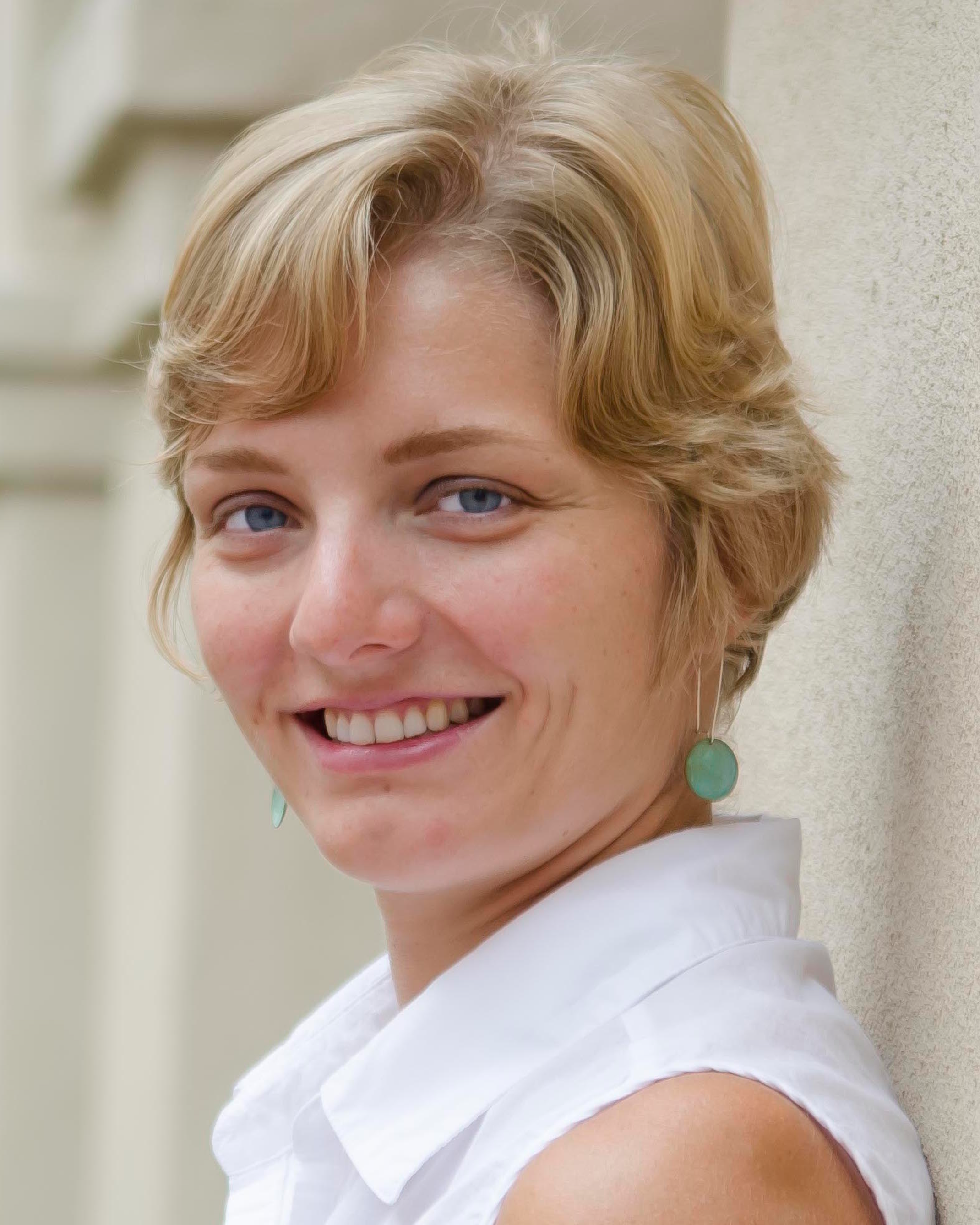}}]{Victoria Kostina}
joined Caltech as an Assistant Professor of Electrical Engineering in the fall of 2014. She holds a Bachelor's degree from Moscow institute of Physics and Technology (2004), where she was affiliated with the Institute for Information Transmission Problems of the Russian Academy of Sciences, a Master's degree from University of Ottawa (2006), and a PhD from Princeton University (2013). She received the Natural Sciences and Engineering Research Council of Canada postgraduate scholarship (2009--2012), the Princeton Electrical Engineering Best Dissertation Award (2013), the Simons-Berkeley research fellowship (2015) and the NSF CAREER award (2017).  Kostina's research spans information theory, coding, control and communications. 
\end{IEEEbiography}

\begin{IEEEbiography}
    [{\includegraphics[width=1in,height=1.25in,clip,keepaspectratio]{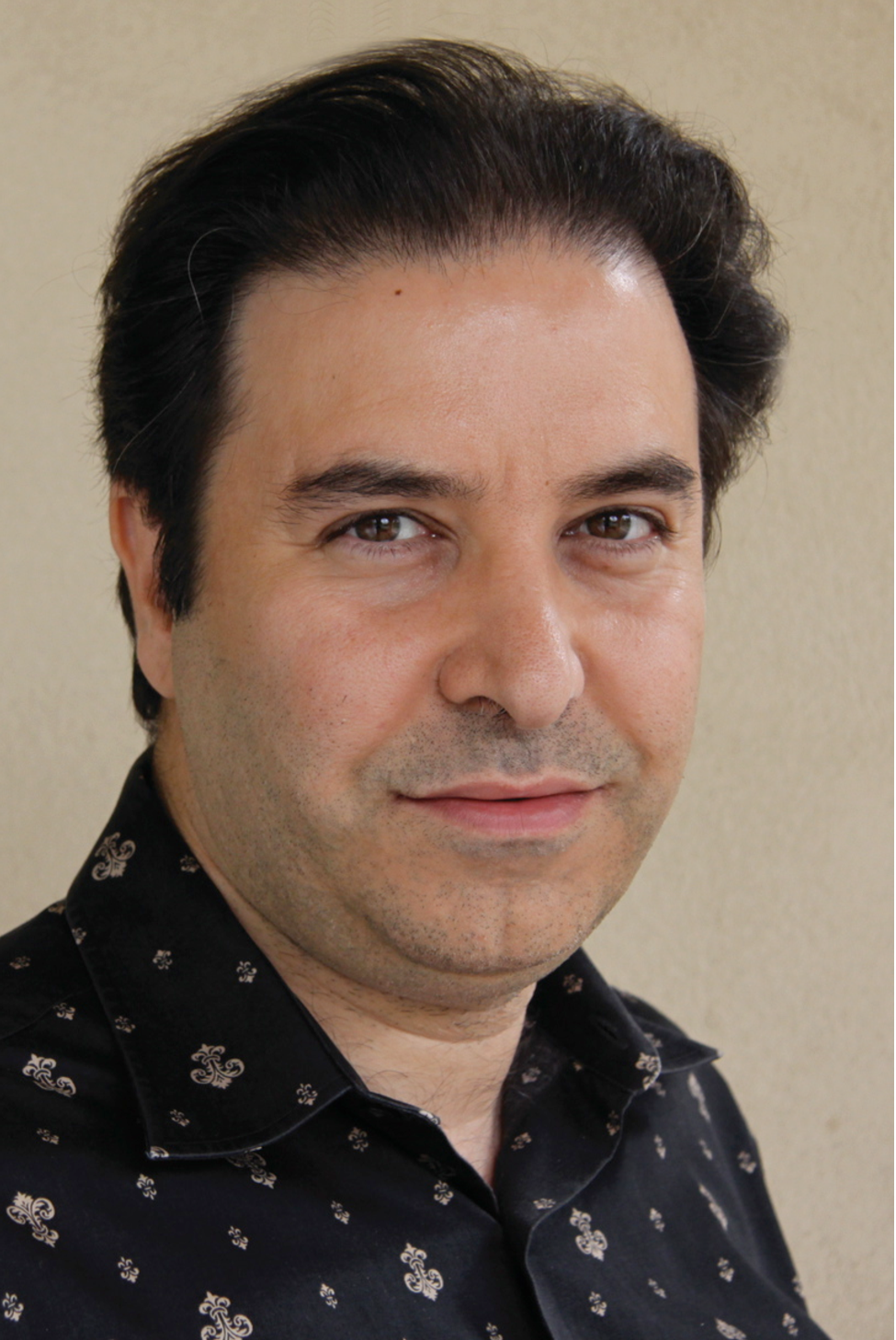}}]{Babak Hassibi}
 was born in Tehran, Iran, in 1967. He received the B.S. degree from the University of Tehran in 1989, and the M.S. and Ph.D. degrees from Stanford University in 1993 and 1996, respectively, all in electrical engineering. 

He has been with the California Institute of Technology since January 2001, where he is currently the Mose and Lilian S. Bohn Professor of Electrical Engineering. From 2013-2016 he was the Gordon M. Binder/Amgen Professor of Electrical Engineering and from 2008-2015 he was Executive Officer of Electrical Engineering, as well as Associate Director of Information Science and Technology. From October 1996 to October 1998 he was a research associate at the Information Systems Laboratory, Stanford University, and from November 1998 to December 2000 he was a Member of the Technical Staff in the Mathematical Sciences Research Center at Bell Laboratories, Murray Hill, NJ. He has also held short-term appointments at Ricoh California Research Center, the Indian Institute of Science, and Linkoping University, Sweden. His research interests include communications and information theory, control and network science, and signal processing and machine learning. He is the coauthor of the books (both with A.H.~Sayed and T.~Kailath) {\em Indefinite Quadratic Estimation and Control: A Unified Approach to H$^2$ and H$^{\infty}$ Theories} (New York: SIAM, 1999) and {\em Linear Estimation} (Englewood Cliffs, NJ: Prentice Hall, 2000). He is a recipient of an Alborz Foundation Fellowship, the 1999 O. Hugo Schuck best paper award of the American Automatic Control Council (with H.~Hindi and S.P.~Boyd), the 2002 National ScienceFoundation Career Award, the 2002 Okawa Foundation Research Grant for Information and Telecommunications, the 2003 David and Lucille Packard Fellowship for Science and Engineering,  the 2003 Presidential Early Career Award for Scientists and Engineers (PECASE), and the 2009 Al-Marai Award for Innovative Research in Communications, and was a participant in the 2004 National Academy of Engineering ``Frontiers in Engineering''program. 

He has been a Guest Editor for the IEEE Transactions on Information Theory special issue on ``space-time transmission, reception, coding and signal processing'' was an Associate Editor for Communications of the IEEE Transactions on Information Theory during 2004-2006, and is currently an Editor for the Journal ``Foundations and Trends in Information and Communication'' and for the IEEE Transactions on Network Science and Engineering. He is an IEEE Information Theory Society Distinguished Lecturer for 2016-2017.
\end{IEEEbiography}

\appendix

This appendix summarizes the tools used in the proofs of \secref{sec:a}. The first is a tool to bound the difference between the differential entropies of two random vectors whose distributions are close to each other. 
\begin{prop}[{\!\cite[Prop. 1]{polyanskiy2015wasserstein}}]
Let $X$ and $Y$ be random vectors with finite second moments. If the density of $X$ is $(c_0, ~c_1)$-regular, then 
\begin{align}
 h(Y) - h(X) 
 &\leq \log e \bigg(\frac {c_1}{2} \left(\E{\| X\|^2}\right)^{\frac 1 2} \notag\\
 &+ \frac {c_1} 2 \left( \E{\| Y\|^2} \right)^{\frac 1 2} + c_0 \bigg)  W(X, Y),
\end{align}
where $W(X, Y)$ is the Wasserstein distance between the distributions of $X$ and $Y$: 
\begin{equation}
 W(X, Y) \triangleq \inf \left( \E{\| X - Y\|^2}\right)^{\frac 1 2},
\end{equation}
where the infimum is over all joint distributions $P_{XY}$ whose marginals are $P_X$ and $P_Y$. 
\label{aprop:pwu}
\end{prop}
The next result helps us establish that the random vectors encountered at each step of the control system operation have regular densities. 
\begin{prop}[{\!\cite[Prop. 3]{polyanskiy2015wasserstein}}]
 If the density of $Z$ is $(c_0, ~c_1)$-regular and $B \pperp Z$, $\|B\| \leq b$ a.s., then that of $B + Z$ is $(c_0 + c_1 b,~ c_1)$-regular.
 \label{aprop:reg}
\end{prop}

 The next result gives an upper bound to the output entropy of lattice quantizers. 
\begin{thm}[Corollary to {\!\cite[Th. 8]{kostina2016lowd}}]
\label{thm:Hqu}
  Suppose that $f_{X}$ is $(c_0, c_1)$-regular. There exists a lattice quantizer $\mathsf q = \mathsf q_{\cC^n}$ such that 
\begin{equation}
 \sup_{x \in \mathbb R^n} \|x - \mathsf q(x)\|^2 \leq d, \label{aeq:covradius}
\end{equation}
and 
\begin{align}
  H\left( q(X) \right) &\leq \min_{\tilde d \leq d} \bigg(
\frac n 2 \log \frac{ N(X)}{\tilde d / n} + \alpha_n + n \log \rho_{\mathcal C_n} \label{aeq:hub} \\
&~+ 2  \tilde d^{\frac1 2} \log e ( c_1 \sqrt{\Var{X}} + c_0 + c_1 \tilde d^{\frac 1 2}) \bigg) , \notag
\end{align}
where $\rho_{\cC_n}$ is the lattice covering efficiency defined in \eqref{eq:rho}, 
\begin{equation}
\alpha_n \triangleq \frac n 2 \log \frac{2 e  }{n} + \log \Gamma \left( \frac n 2 + 1\right),\label{aeq:alphan}
\end{equation}
and $\Gamma(\cdot)$ is the Gamma function.
\end{thm}

The leading term in \eqref{aeq:hub} is Shannon's lower bound (the functional inverse of \eqref{eq:slb}). The contribution of the remaining terms becomes negligible if $n$ is large and $d$ is small.  
Indeed, by Stirling's approximation, as $n \to \infty$, 
\begin{equation}
\alpha_n = \frac 1 2 \log n + \bigo{1}.  
\end{equation}
On the other hand, Rogers \citeapx[Theorem 5.9]{rogers1964packing} showed that for each $n \geq 3$, there exists an $n$-dimensional lattice $\cC_n$  with covering efficiency
\begin{equation}
n \log \rho_{\cC_n} \leq \log_2 \sqrt{2 \pi e} \left( \log n + \log \log n + c\right), \label{aeq:rogers}
\end{equation}
where $c$ is a constant.
Therefore, the terms $n \log \rho_{\mathcal C_n}$ and $\alpha_n$ are logarithmic in $n$, so in high dimension their contribution becomes negligible compared to the first term in \eqref{aeq:hub}. 

In low dimension, the contribution of these terms can be computed as follows. The thinnest lattice covering is known in dimensions 1 to 23 is Voronoi's principal lattice of the first type \citeapx{conway2013sphere} $\left(A_n^*\right)$, which has covering efficiency
\begin{equation}
 \rho_{A_n^*} = \frac{\pi^{\frac 1 2}(n + 1)^{\frac 1 {2n}} }{\left(\Gamma\left( \frac n 2 + 1\right)\right)^\frac{1}{n}}  \sqrt{ \frac {n(n + 2)}{12 (n + 1)} }.
\end{equation}
$A_n^*$ is proven to be the thinnest lattice covering possible in dimensions $n = 1, 2, \ldots, 5$. For $A_n^*$-based lattice quantizer, we can compute the constant appearing in \eqref{aeq:hub} as
\begin{equation}
 \log \rho_{A_n^*} + \frac {\alpha_n} n  = \frac 1 2 \log { \frac {2 \pi e(n + 2)}{12 (n + 1)^{1 - {\frac 1 n}}} }.
\end{equation}

\bibliographystyleapx{IEEEtran}
\bibliographyapx{../../Bibliography/rateDistortion,../../Bibliography/control,../../Bibliography/vk}

\end{document}